\documentclass[a4paper,11pt]{article}
\usepackage[margin=2.5cm]{geometry}

\usepackage{fullpage}
\usepackage{indentfirst}
\usepackage{microtype}
\usepackage{float}
\usepackage{xcolor}
\usepackage{listings}
\usepackage[inline]{enumitem}
\usepackage[utf8]{inputenc}
\usepackage[english]{babel}
\usepackage{amsthm}
\usepackage{amsmath,mathtools,mathrsfs}
\usepackage{amssymb}
\usepackage[T1]{fontenc}
\usepackage{inconsolata}
\usepackage{color}
\usepackage{graphicx}
\usepackage[backend=biber,maxbibnames=99,style=numeric,sorting=nyt]{biblatex}
\usepackage{csquotes}
\allowdisplaybreaks
\usepackage{float}
\usepackage{hyperref}
\usepackage[ruled,linesnumbered,noend]{algorithm2e}
\usepackage[T1]{fontenc}
\usepackage{lmodern}
\usepackage{tikz}
\SetNlSty{textbf}{}{:}
\SetKwInput{kwInit}{Init}

\SetKwIF{If}{ElseIf}{Else}{if}{\hspace{-0.13cm}then}{else if}{else}{end if}%

\newlength\FirstColLength
\newtheorem{theorem}{Theorem}[section]
\newtheorem{lemma}[theorem]{Lemma}
\newtheorem{corollary}[theorem]{Corollary}
\newtheorem{observation}[theorem]{Observation}

\newtheorem{definition}[theorem]{Definition}

\title{Exactly Optimal Deterministic Radio Broadcasting with Collision Detection}
\date{}
\author{
    Koko Nanahji \\
    University of Toronto \\
    Toronto, Canada \\
    koko.nanahji@mail.utoronto.ca \\
}

\begin{document}    
    \pagenumbering{arabic}
    \maketitle
%%%%%%%%%%%%%%%%%%%%%%%%%%%%%%%%%%%%%%%%%%%%%%%%%%%%%%%

\section*{Abstract}
    We consider the broadcast problem in synchronous radio broadcast models with collision detection.
    One node of the network is given a message that must be learned by all nodes in the network. 
    We provide a deterministic algorithm that works on the beeping model, which is a restricted version of the radio broadcast model with collision detection.
    This algorithm improves on the round complexity of previous algorithms.
    We prove an exactly matching lower bound in the radio broadcast model with collision detection. 
    This shows that the extra power provided by the radio broadcast model with collision detection does not help improve the round 
    complexity.

\section{Introduction}
        Broadcast is a fundamental problem in distributed computing that is frequently used as a building block in other problems. 
        We study the broadcast problem in synchronous radio networks, modeled as undirected connected graphs,
        where nodes do not have identifiers and, initially, have no knowledge about the network.
        In this problem, one node of the network is given a message that must be learned by all nodes in the network. 
            
        In the radio broadcast model with collision detection, a listening node receives the entire message of its neighbour if that neighbour is its only neighbour which transmits that round. If more than one of its neighbours transmit, then it receives a collision signal.
        The beeping model is a restricted version in which nodes can only transmit the collision signal.
        
        Chlebus, G{\k{a}}sieniec, Gibbons, Pelc, and Rytter
        \cite{Bogdan2002}, provide a deterministic broadcasting algorithm that works on the beeping model.
        This algorithm takes $\Theta( D \cdot \log{\mu})$
        rounds to broadcast 
        % This algorithm takes $(D+1) \cdot (2 \cdot \lceil \log_2{\mu}\rceil+4)$ rounds to broadcast 
        any message $\mu \in \mathbb{Z}^{+}$, 
        where $D$ is the source eccentricity of the network. 
        Czumaj and Davies \cite{CZUMAJ2019} and, independently, Hounkanli and Pelc \cite{Pelc2015}, provide an asymptotically optimal deterministic broadcasting algorithm in the beeping model that takes $D+6\cdot \lceil \log_2{\mu}\rceil+11$ rounds to broadcast any message $\mu \in \mathbb{Z}^{+}$. 
        They also prove a lower bound of $\Omega(D+\log{\mu})$ rounds.
        
        In this paper, we provide a deterministic algorithm to broadcast any value from a predefined set of values that has exactly optimal round complexity.
        In particular, in Section \ref{subsection:algorithm}, 
        we provide a deterministic algorithm to broadcast any value from set $\{1,\ldots, m\}$ that works on the beeping model and 
        takes $D+r(m)$ rounds, where $r(m) \leq 2\cdot \lceil \log_2{m}\rceil +2$. 
        in Section \ref{section:lower_bound}, we prove that 
        $D + r(m)$ is a lower bound
        %our algorithm has exactly optimal round complexity 
        in the radio broadcast model with collision detection. 
        This shows that 
        allowing nodes to send arbitrarily long messages instead of just a collision signal
        does not help improve the round complexity.

\section{Model}

         There are a number of variants of the synchronous radio broadcast model that are considered. The three most relevant to this paper are the following: 
        % have been studied:
        
        \begin{itemize}
        \item
            {\bf Radio broadcast model without collision detection.}
            In this model, the nodes communicate in synchronous rounds. 
            At each round, a node can idle, listen, or transmit the same value to all its neighbours. 
            If a node transmits or idles, it gets no feedback from the environment at that round. 
            If a node listens and none of its neighbors transmit
            or at least two of its neighbors transmit, it receives nothing.
            If a node listens and exactly one of its neighbors transmits, then it receives the value transmitted by that neighbour.
            % If , a listener receives nothing. 
            
   \item    
            {\bf Radio broadcast model with collision detection.}
            This model is same as the previous model except that, if at least two neighbors of a node transmit, it receives a collision signal.
            Thus, a node that listens can distinguish between none of its neighbours transmitting and at least two of its neighbours transmitting. 
        
   \item
            {\bf Beeping model.}
            This is a special case of the radio broadcast model with collision detection, in which the only message a node can transmit is the collision signal (also called a beep).
            Thus, in the beeping model, when a node listens, it receives a beep if at least one of its neighbours transmits at that round.
        \end{itemize}
        \medskip
        
        We assume that, initially, the nodes are indistinguishable and have no knowledge about the network.
        Each node starts by listening (without recording any information) up to and including the first round in which at least one of its neighbours transmits. 
        We say that a node \textit{wakes up} at some round $t$, if at least one of its neighbours transmits for the first time at round $t$.
        A node that wakes up at round $t$ starts executing the given algorithm at round $t+1$.
        
        In the broadcast problem, an external source wakes up one of the nodes in the network by giving it a value from a finite set of possible values. This node is called the source node.
        %The goal is to have all nodes in the network learn the value of the message.
        We define round $1$ of an execution to be the round immediately following the round in which the source node is woken up. 
        The nodes at distance $\ell$ from the source node are said to be at level $\ell$. The maximum level of any node in the network is called the source eccentricity. 
        
\section{Related Work}
            Early approaches to solve the broadcast problem 
            used randomization.
            Bar-Yehuda, Goldreich, and Itai \cite{BARYEHUDA1992104} were first to provide a randomized algorithm for the broadcasting problem in the radio broadcast model without collision detection.
            They assumed that nodes are anonymous (i.e., they do not have any identifiers) and they are given the size of the network and the maximum number of neighbours of any node in the network.
            Their algorithm takes $\mathcal{O}(D \cdot \log{n}+ \log^2{n})$ rounds with high probability, where $n$ is the size of the network and $D$ is the source eccentricity of the network. 
            Later, Czumaj and Rytter \cite{Czumaj2003} and, independently, Kowalski and Pelc \cite{pelc2003} provided a randomized broadcasting algorithm 
            where 
            nodes have distinct identifiers and they are given the size of the network.
            Their algorithm takes expected $\mathcal{O}(D\log{\frac{n}{D}}+ \log^2{n})$ rounds.
            %This is asymptotically optimal, since
            Kushilevitz and Mansour \cite{Mansour1993} proved an expected $\Omega(D\log{\frac{n}{D}})$ round lower bound, when nodes only know the size and the diameter of the network.
            Alon, Bar-Noy, Linial, and Peleg \cite{Alon1991} proved an $\Omega(\log^2{n})$ round lower bound, even if nodes have distinct identifiers and every node knows the entire network.
            Thus, the expected round complexity of the broadcast problem
            without collision detection is in $\Theta(D\log{\frac{n}{D}}+ \log^2{n})$ 
            when nodes have distinct identifiers
            and are given the size of the network.
            % This lower bound holds for randomized algorithms as well. The reason for this is not stated but I think is because  it assumes nodes have IDs and complete knowledge of the network.
            
            Ghaffari, Haeupler, Khabbazian \cite{ghaffari2015randomized} studied randomization in the radio broadcast model with collision detection.
            They provided a broadcasting algorithm that takes $\mathcal{O}(D+\log^6{n})$ rounds, with high probability, in which the nodes are given the size of the network and the source eccentricity of the network.
            They also presented a broadcasting algorithm that takes $\mathcal{O}(D+\log^2{n})$ rounds, with high probability, in which the nodes know the entire network topology. 
      
        Deterministic approaches have also been studied.
        Most of the papers about the deterministic broadcast problem are for
        the radio broadcast model without collision detection.
                Initially, researchers considered the case in which each node has a distinct identifier.
                Some papers assume that nodes are not given any information about the network, except possibly the size of the network $n$, the maximum number of neighbours of any node in the network $\Delta$, and the source eccentricity of the network $D$.
                There are many broadcasting
                algorithms for this setting
                \cite{Bogdan2002,
                Bogdan2000,
                Chrobak2002,
                CLEMENTI2003}.
                The fastest such algorithm was provided by Czumaj and Davies \cite{Czumaj2018}, which takes $\mathcal{O}( n\log{D} \log{\log \frac{D\Delta}{n}})$ rounds.

                Other papers consider algorithms designed with knowledge 
                of the entire network topology.
                They also assume that nodes have distinct identifiers.
                Chlamtac \cite{wave_expansion_approach} provided a deterministic algorithm for the broadcast problem that takes $\mathcal{O}(D \cdot \log^2{n})$ rounds.
                Kowalski and Pelc \cite{Optimal_Deterministic_Broadcasting_in_Known_Topology_Radio_Networks} provided an optimal  deterministic algorithm, which takes $\mathcal{O}(D + \log^2{n})$ rounds.
            
                Another approach was  
                to carefully assign short labels to the nodes of the network, instead of assuming they have distinct identifiers.
                % Ellen, Gorain, Miller, and Pelc \cite{Ellen2019} observed that, without identifiers and without collision detection, it is impossible to solve the broadcast problem deterministically in some radio networks because it is impossible to break symmetry.
                Ellen, Gorain, Miller, and Pelc \cite{Ellen2019}
                %They also 
                showed that
                broadcast can be done in any radio broadcast network after assigning $2$-bit labels to each node.
                Their algorithm takes $\Theta(n)$ rounds.
                They also showed that $1$-bit labels were sufficient for a class of networks.
                Gewu, Potop-Butucaru, Rabie \cite{Bu2021} proved that $1$-bit labels are sufficient to solve the broadcast problem for a larger class of networks.
                Ellen and Gilbert \cite{Ellen2020} provide an algorithm that uses $3$-bit labels and takes $\mathcal{O}(D \log^2 n)$ rounds and an algorithm that uses $4$-bit labels and takes $\mathcal{O}(\sqrt{Dn})$ rounds.
                Ellen, Gorain, Miller, and Pelc \cite{Ellen2019} observed that, without identifiers and without collision detection, it is impossible to solve the broadcast problem deterministically in some radio networks because it is impossible to break symmetry.
                
                \medskip
                
                The known broadcasting algorithms for the beeping model
                \cite{Bogdan2002,CZUMAJ2019,Pelc2015} 
                do not assume that nodes have identifiers.
                Some of the techniques used in these algorithms are 
                closely related to techniques used in our paper, so we describe them in more detail.
                % would like to give more details about them.
                
                In the algorithm by Chlebus, G{\k{a}}sieniec, Gibbons, Pelc, and Rytter \cite{Bogdan2002},
                all bits of the message are transmitted
                to nodes at the same level before
                any bits of this message are transmitted
                to nodes that are at higher levels. 
                The algorithm to broadcast message $\mu$ is executed in phases of $\mathcal{O}(\log_2{\mu})$ rounds. 
                At each phase $1 \leq i \leq D$, the nodes at level $i-1$  transmit the message (bit by bit) to the nodes at level $i$. 
                Thus, this algorithm takes $\mathcal(D \cdot \log_2{\mu})$ rounds.
                
                The algorithms presented by Czumaj and Davies \cite{CZUMAJ2019} and Hounkanli and Pelc \cite{Pelc2015} are very similar. 
                Both rely on the \textit{Beep Waves} method, introduced by Ghaffari and Haeupler \cite{ghaffari2014near}, to relay a message using beeps. 
                The idea is to broadcast the bits of the message level by level in a pipelined fashion.
                At every third round, starting with round $1$,
                the source node announces a bit of the message. 
                The source node beeps at that round if and only if the bit is $1$.
                When any other node learns a bit of the message, it conveys the value of that bit to its neighbours at the next level during the following round. 
                Then it idles for one round,
                ignoring the information those neighbours send.
                In particular, nodes at level $\ell \geq 1$ learn the $i^{th}$ bit of the message at round $\ell+3(i-1)$.
                During the execution of Beep Waves, a beep of the source node propagates like a wave throughout the network, spreading to the next level at each round. 
                
                In all three algorithms, a non-source node learns the message  bit by bit over several rounds and, hence, it needs a mechanism to detect the end of the message. If the set of messages is $\{1,\ldots,m\}$, it suffices for the source node to send
                the $\lceil \log_2 m \rceil$ bits of a binary encoding of the message.
                However, if the message can be any positive integer, the encoding of the message must be self-delimiting. 
                For example, Czumaj and Davies \cite{CZUMAJ2019} start the encoding with $10$, then
                duplicate each bit in the binary representation of the message, and end the encoding with $10$.
                Hence, binary representation of $5$ becomes $1011001110$ in their encoding.

\section{Algorithm}   \label{subsection:algorithm}
    First, we describe an algorithm in the beeping model in which the source node can broadcast any value from a prefix-free set of binary strings. 
    The algorithm uses a variant of the \textit{Beep Waves} method.
    However, the pipelining operates slightly differently, where transmitting a $0$ takes fewer rounds.  
    Later, we design a new encoding scheme 
    which takes this into account. The combination of the algorithm and the encoding scheme is shown to be optimal in the next section.

    % The algorithm can be altered to work on the radio broadcast model with collision detection as follows. Modify the algorithm such that nodes transmit instead of beeping, and when a node receives a message or a collision signal it acts as though it received a beep. 
    
    Given a binary string $x$, let $|x|$ denote the length of $x$, let $x_i$ denote the $i^{th}$ character of $x$ for $1 \leq i \leq |x|$, and let $\textit{pre}_j(x)$ be the prefix of length $j$ of $x$, where $0 \leq j \leq |x|$. 
    
    Let $S$ be a prefix-free set of binary strings (i.e. no element of $S$ is a prefix of another element of $S$).
    First, we provide a detailed description of the algorithm to broadcast any string $s \in S$.
    
    \begin{center}
    \SetAlgoNlRelativeSize{0}
    \begin{algorithm}[H]
        \SetAlgoLined
        \DontPrintSemicolon 
        \SetAlgoNoLine
        % \footnotesize
        Beep \;
        Listen \;
        \lIf{received nothing \hspace{0.13cm}}{Terminate}
        Idle \;
        \For{$i$ from $1$ to $|s|-1$}{
        \iffalse    
            \If{$s_i = 1$}{
                Beep \;    
                Idle \;
                Idle
            }
            \lElse{Idle}
        \fi
            \If(Beep){$s_i = 1$ \\}{
                \hspace{0.33cm} Idle \;
                \hspace{0.33cm} Idle
            }
            \lElse{\hspace{0.2cm}Idle}
        }
        \lIf{$s_{|s|} = 1$ \hspace{0.13cm}}{Beep}
        Terminate
        \caption{Algorithm for the Source Node To Broadcast $s$ After Being Woken Up}
        \label{algorithm:broadcast_K_values_source}
    \end{algorithm}
    \end{center}
    
    \begin{center}
    \SetAlgoNlRelativeSize{0}
    \begin{algorithm}[H]
        \SetAlgoLined
        \DontPrintSemicolon 
        \SetAlgoNoLine
        % \footnotesize
        $\textit{acc}_v \longleftarrow \epsilon$\;
        Beep \;
        Listen \;
        \leIf{received nothing \hspace{0.13cm}}{$\textit{terminal}_v \longleftarrow \textit{True}$ \hspace{0.01cm}} {$\textit{terminal}_v \longleftarrow \textit{False}$}
        \While{$\textit{acc}_v \notin S$} {
            Listen \;
            \If($\textit{acc}_v \longleftarrow \textit{acc}_v  \cdot 1$){received signal \\ }
               {
                \hspace{0.35cm}
                \If(Beep){$\textit{acc}_v \notin S$ \\ \hspace{0.45cm} }{
                    \hspace{0.77cm} Idle 
                }
                \hspace{0.35cm}
                \ElseIf(Beep){$\textit{terminal}_v = \textit{False}$ \hspace{0.13cm}}{}
            }
            \lElse{$\textit{acc}_v \longleftarrow \textit{acc}_v  \cdot 0$}
        }
        Terminate
        \caption{Algorithm for a Non-Source Node $v$ After Being Woken Up}
        \label{algorithm:broadcast_K_values_non_source}
    \end{algorithm}
    \end{center}

    The pseudocode for the source node is presented in Algorithm \ref{algorithm:broadcast_K_values_source}. 
    In the first round after it is woken up, the source node beeps. 
    In the second round, the source node listens. 
    If it receives nothing, then it terminates. Otherwise, it idles for one round and then, 
    for $1 \leq i \leq |s|-1$, it does one of the following. 
    If $s_i=0$, then the source node idles. 
    If $s_i=1$, then the source node beeps and idles for two rounds. 
    Finally, if the last character of $s$ is $1$, then the source node beeps. 
    
    The pseudocode for an arbitrary non-source node, $v$, is presented in Algorithm  \ref{algorithm:broadcast_K_values_non_source}. 
    % Let $v$ be an arbitrary non-source node. 
    Its local variable $\textit{acc}_v$
    is used to accumulate the characters of $s$ and is initially set to the empty string.
    Once $v$ is woken up, it beeps in the next round and listens in the following round. 
    It receives a beep in this round if and only if 
    it is connected to a node at a higher level.
    This information is recorded in its local variable $\textit{terminal}_v$.
    %If $v$ receives nothing, then it records that it is a terminal node. Otherwise, $v$ records that it is not a terminal node. 
    %We will show that a node is terminal if and only if it is not connected to a node at a higher level.
    
    While node $v$ has not yet received the entire message,
    %$\textit{acc}_v$ is a proper prefix of strings in $S$,
    it repeatedly executes phases.
    % Then $v$ executes phases repeatedly until $\textit{acc}_v$ is an entire string of $S$. 
    % Then $v$ executes the following phase repeatedly until $\textit{acc}_v \in S$. 
    Each phase starts by a listen round, in which it learns the next bit of the message. Specifically, receiving nothing means this bit is 0 and receiving a beep 
    means this bit is 1. In the latter case, it beeps at the next round to relay this information to the nodes at the next level and then idles for one round.
        In the last phase, immediately after it has learned that the last bit of the message is 1,
    a node that is not connected to a node at a higher level can immediately terminate. If it is
    connected to a node at a higher level, then it can terminate
    after beeping.
    %to relay the final 1 bit of the message. 
    %then it spends two more rounds to relay this information to the nodes at the next level.
    %If $v$ receives nothing at that round, it appends $0$ to $\textit{acc}_v$. Otherwise
    %If $v$ receives a signal at that round,
    \medskip
    
    Now, we prove the correctness and bound the round complexity of this algorithm when broadcasting message $s$.
    % First, we identify the round in which each non-source node wakes up.
    Since each node beeps after waking up, 
    the nodes at level $\ell$ wake up at round $\ell$.
    Moreover, if there are nodes at level 1, then the source node receives a signal at round 2 and, hence, continues with the rest of its algorithm. Likewise, a non-source node $v$ sets $\textit{terminal}_v$ to $\textit{True}$ if and only if it is connected to a node at a higher level.
    
    Let $C(x)$ be the number of $0$'s in the binary string $x$ plus three times the number of $1$'s in $x$. 
    % Next, we define function $C$ and show that when the source node broadcasts value $s$, all nodes are terminated at the end of round $D+C(pre_{|s|-1}(s))+3$.
    % $C(x)$ is equal to the number of $0$'s in $x$ plus three times the number of $1$'s in $x$. 
    In particular, $C(x)=0$, if $x$ is the empty string. 
    
    Observe that, the source node executes $3$ rounds before executing the first iteration of the loop. 
    If $s_i = 0$, then the $i^{th}$ iteration of the loop takes $C(s_i)=1$ round. 
    If $s_i = 1$, then the $i^{th}$ iteration of the loop takes $C(s_i)=3$ rounds. 
    Therefore, $3+C(pre_{i-1}(s))+1$ is the first round of the $i^{th}$ iteration, for $1 \leq i \leq |s|-1$.
    \begin{observation} \label{observation:k_values:algorithm:source_node_beep_rounds}
        Suppose the source node does not terminate at the end of round $2$.
        Then, for all $1 \leq i \leq |s|$, then the source node beeps at round $3+C(pre_{i-1}(s))+1$ if and only if $s_i=1$.
        Furthermore, the source node terminates at round $3+C(pre_{|s|-1}(s))+s_{|s|}$.
    \end{observation}  
    
    % Note that if $x'$ is a prefix of $x$ then $C(x') \leq C(x) $, and thus $C$ is a strictly increasing function.
    We will show that when $D >0$, all nodes terminate by the end of round $D+C(pre_{|s|-1}(s))+3$.
    
    In the next two results, we identify the rounds in which each node learns each bit of the message. 
    % we show that, for all $1 \leq \ell \leq D$,  nodes at level $\ell$ learn the value of $s$ at the end of round $\ell+C(pre_{|s|-1}(s))+3$.
    First, we identify the round in which each node learns the first bit of $s$.
    Then, we generalize it to all bits of $s$.
    \begin{lemma} \label{lemma:k_values:algorithm:learn_first_character}
        For all $1 \leq \ell \leq D$, each node $v$ at level $\ell$ appends a character to $\textit{acc}_v$ for the first time at round $\ell+3$ and the appended character is $s_1$.
    \end{lemma}
    \begin{proof}
        The source node beeps at round $4$ if and only if $s_1=1$.
        Each node $u$ at level $1$ wakes up at round $1$ and appends a character to $\textit{acc}_u$ for the first time at round $4$.
        The nodes at level $2$ wake up at round $2$ and listen at round $4$.
        Therefore, nodes at level $1$ receive a signal at round $4$ if and only if $s_1=1$.
        Since $u$ appends $1$ to $\textit{acc}_u$ if it receives a signal and appends $0$ if it receives nothing,
        $u$ appends $s_1$ to $\textit{acc}_u$ at round $4$.
        % Hence, the claim holds for $\ell=1$.
        
        % Induction step: 
        Let $2 \leq \ell \leq D$ and assume the claim is true for $\ell-1$. 
        In particular, each node $w$ at level $\ell-1$ appends $s_1$ to $\textit{acc}_w$ at round $\ell+2$. 
        From the pseudocode, if $\textit{terminal}_w=\textit{False}$, then $w$ beeps at round $\ell+3$ if and only if $s_1=1$. 
        Since nodes at level $\ell$ wake up at round $\ell$, 
        each node $v$ at level $\ell$ listens and appends $0$ or $1$ to $\textit{acc}_v$ for the first time at round $\ell+3$.
        If $\ell < D$, the nodes at level $\ell+1$ wake up at round $\ell+1$ and listen at round $\ell+3$.
        Therefore, $v$ receives a signal at round $\ell+3$ if and only if $s_1=1$. 
        Since $v$ appends $1$ to $\textit{acc}_v$ if it receives a signal and appends $0$ if it receives nothing,
        $v$ appends $s_1$ to $\textit{acc}_v$ at round $\ell+3$.
    \end{proof}
 
    \begin{lemma} \label{lemma:k_values:algorithm:main_lemma}
        For all $1 \leq j \leq |s|$, each node $v$ at level $1 \leq \ell \leq D$ appends a character to $\textit{acc}_v$ for the $j^{th}$ time at round $\ell+C(pre_{j-1}(s))+3$ and the appended character is $s_j$.
    \end{lemma}
    \begin{proof}
        We will prove the claim by induction on $1 \leq j \leq |s|$ and $1 \leq \ell \leq D$.
        % Base Case:  
        Note that, if $j=1$, then $C(pre_{j-1}(s))=0$. 
        Thus, by Lemma \ref{lemma:k_values:algorithm:learn_first_character}, the claim holds for all $1 \leq \ell \leq D$ when $j=1$.
        
        % Induction step: 
        Let $2 \leq j' \leq |s|$ and $1 \leq \ell' \leq D$. 
        Assume the claim holds for $1 \leq \ell \leq D$ when $j=j'-1$ and for $1 \leq \ell \leq \ell'-1$ when $j=j'$. 
        
        % First, we show that nodes at level $\ell'$ append a value to $\textit{acc}$ for the $j^{th}$ time at round $\ell+C(pre_{j'-1}(s))+3$.
        Let $v$ be a node at level $\ell'$.
        It follows from the induction hypothesis that
        $v$ appends $s_{j'-1}$ to $\textit{acc}_v$ at round $\ell'+C(pre_{j'-2}(s))+3$ and the value of $\textit{acc}_v$ becomes $pre_{j'-1}(s)$.
        Since $S$ is prefix-free, $pre_{j'-1}(s) \notin S$ at the end of round $\ell'+C(pre_{j'-2}(s))+3$.
        From the pseudocode, 
        if $s_{j'-1}=0$, then $C(s_{j'-1})=1$ round later,
        $v$ listens and appends a character to $\textit{acc}_v$.
        If $s_{j'-1}=1$, then  $C(s_{j'-1})=3$ rounds later, $v$ next listens and appends a  character to $\textit{acc}_v$.
        Thus, $v$ appends the next character to $\textit{acc}_v$ at round $\ell'+C(pre_{j'-2}(s))+3+C(s_{j'-1})=\ell'+C(pre_{j'-1}(s))+3$.
        
        % Second, we show that nodes at level $\ell+1$ (if any) do not beep at round $\ell+C(pre_{j'-1}(s))+3$.
        Let $u$ be a node at level $\ell'+1$.
        By the induction hypothesis, $u$ appends $s_{j'-1}$ to $\textit{acc}_u$ at round $\ell'+1+C(pre_{j'-2}(s))+3$.
        Thus,
        if $s_{j'-1}=0$, then $u$ listens at round $\ell'+1+C(pre_{j'-2}(s))+3=\ell'+C(pre_{j'-1}(s))+3$.
        If $s_{j'-1}=1$, 
        then, $u$
        idles at round 
        $2+\ell'+1+C(pre_{j'-2}(s))+3=\ell'+C(pre_{j'-1}(s))+3$.
        Hence, nodes at level $\ell'+1$ do not beep at round $\ell'+C(pre_{j'-1}(s))+3$.
        
        % Finally, we show that non-terminal nodes at level $\ell-1$ beep at round $\ell+C(pre_{j'-1}(s))+3$ if and only if $s_{j'}=1$.
        If $\ell' = 1$, then, 
        by Observation \ref{observation:k_values:algorithm:source_node_beep_rounds}, 
        % From the pseudocode, the source node executes $3$ rounds before executing the first iteration of the loop. 
        % Furthermore, the $i^{th}$ iteration of the loop takes $C(s_i)$ rounds.
        % Thus, 
        the source node beeps at round $3+C(pre_{j'-1}(s))+1=\ell'+C(pre_{j'-1}(s))+3$ if and only if $s_{j'}=1$.
        So, suppose $\ell' \geq 2$. 
        By the induction hypothesis, each node $w$ at level $\ell'-1$ appends $s_{j'}$ to $\textit{acc}_w$ at round $\ell'-1+C(pre_{j'-1}(s))+3$. 
        From the pseudocode, if $\textit{terminal}_w=\textit{False}$,
        $w$ beeps at round $\ell'+C(pre_{j'-1}(s))+3$ if and only if $s_{j'}=1$.
        Therefore, $v$ receives a signal at round $\ell'+C(pre_{j'-1}(s))+3$ if and only if $s_{j'}=1$. 
        Since $v$ appends $1$ to $\textit{acc}_u$ if it receives a signal and appends $0$ if it receives nothing,
        $v$ appends $s_{j'}$ to $\textit{acc}_v$ at round $\ell'+C(pre_{j'-1}(s))+3$. 
    \end{proof}

    Now, we show that all nodes learn $s$ and terminate by the end of round $D+C(pre_{|s|-1}(s))+3$.
    \begin{lemma} \label{lemma:k_values:algorithm:termination}
        For all $1 \leq \ell \leq D$, nodes at level $\ell$ learn the value of $s$ at round $\ell+C(pre_{j-1}(s))+3$.
        Each node $v$ at level $\ell$ terminates at round $\ell+C(pre_{|s|-1]}(s))+3+s_{|s|}$ if $\textit{terminal}_w=\textit{False}$ and terminates at round $\ell+C(pre_{|s|-1}(s))+3$ if $\textit{terminal}_w=\textit{True}$.
    \end{lemma}
    \begin{proof}
        Let $1 \leq \ell \leq D$ and let $v$ be a node at level $\ell$.
        By Lemma \ref{lemma:k_values:algorithm:main_lemma},
        $\textit{acc}_v=s$ at the end of round $\ell+C(pre_{j-1}(s))+3$.
        If $s_{|s|}=0$ or $\textit{terminal}_v=\textit{True}$, then $v$ terminates at the end of this round. 
        Otherwise, $v$ terminates one round later. 
    \end{proof}

    \medskip
    
    Let $m\geq 2$.
    We now explain how to use the algorithm in Figures \ref{algorithm:broadcast_K_values_source} and \ref{algorithm:broadcast_K_values_non_source} to broadcast any value from $\{1,\ldots, m\}$. 
    Since each value of an arbitrary set of size $m$ can be mapped to a value in $\{1,\ldots, m\}$, this method can be used to broadcast a value from an arbitrary set of size $m$.
    
    First, we recursively construct a prefix-free set of binary strings $W_i$ such that, for all $i\geq 3$ and all $w \in W_i$, $C(pre_{|w|-1}(w)) \leq i-3$. 
    %Then, we describe the modifications to the algorithm in Figures  \ref{algorithm:broadcast_K_values_source} and \ref{algorithm:broadcast_K_values_non_source} to broadcast a value from $\{1,\ldots, m\}$.
    Let $W_{0} = \{0\}$, $W_{1} = W_{2} = \{1\}$, $W_{3} = \{0,1\}$, $W_{4} = \{00,01,1\}$ and $W_{5} = \{000,001,01,1\}$. 
    For all $i \geq 6$, $W_{i}$ consists of the strings in $W_{i-1}$ prepended by $0$ (denoted by $0 \cdot W_{i-1}$) and the strings in $W_{i-3}$ prepended by $1$ (denoted by $1 \cdot W_{i-3}$). 
    % the following holds.
    % \begin{equation*} 
    %     w_{i} = \begin{cases}
    %                 1& i=0,1,2\\
    %                 w_{i-1}+ w_{i-3}& i\geq 3
    %              \end{cases} 
    % \end{equation*}
    For all $i\geq0$, let $w_i=|W_i|$.
    Then,
    $w_{0}=w_{1}=w_{2}=1$ and, for all $i\geq 3$, $w_{i}=w_{i-1}+ w_{i-3}$.
    Note that, $w_{0}, w_{1}, w_{2}, \ldots$ is Narayana's cows sequence \cite{narayanaBook}.
    % According to Benoit Cloitre \cite{oeis},  % add reference 
    % $w_{i}=\left\lfloor d c^{i}+\frac{1}{2}\right\rfloor$,
    % where $c\approx 1.4655$, and $d \approx 0.6115$. 
    % Prove prefix code
    We show that $W_i$ is prefix-free.
    \begin{lemma}
        For all $i\geq 0$, there is no string in $W_i$ that is a prefix of any other string in $w_{i}$.
    \end{lemma}
    \begin{proof}
        By inspection, the claim holds for $0\leq i \leq 5$.
        % Induction step:
        Let $i\geq 6$ and assume the claim is true for $W_{i-1}$ and $W_{i-3}$.
        Hence, no string in $0 \cdot W_{i-1}$ is a prefix of any other string in $0\cdot W_{i-1}$. 
        Similarly, no string in $1 \cdot W_{i-3}$ is a prefix of any other string in $1 \cdot W_{i-3}$.
        Since $W_i = 0 \cdot W_{i-1} \cup 1 \cdot W_{i-3}$, no string in $W_{i}$ is a prefix of any other string in $W_{i}$.
    \end{proof}
    
    Next, we describe the relationship between $w \in W_i$ and the function $C$ for all $i \geq 3$.
    \begin{lemma} \label{lemma:k_values:algorithm:encoding_cost_upper_bound}
        For all $i\geq 3$, all $w \in W_i$, $C(pre_{|w|-1}(w)) \leq i-3$.
    \end{lemma}
    \begin{proof}
        By inspection, the claim holds for $3\leq i \leq 5$.
        % Induction step:
        Let $i\geq 6$ and assume the claim is true for $W_{i-1}$ and $W_{i-3}$.
        Let $w \in W_{i}$ be arbitrary. 
        If $w=0\cdot w'$,
        % , where $w'$ is the suffix of $w$ of length $|w|-1$.
        then $w' \in W_{i-1}$. 
        By the induction hypothesis, $C(pre_{|w'|-1}(w')) \leq i-4$. Hence, $C(pre_{|w|-1}(w)) \leq i-3$.
        If $w=1\cdot w'$, 
        then $w' \in W_{i-3}$. 
        By the induction hypothesis, $C(pre_{|w'|-1}(w')) \leq i-6$.
        Hence, $C(pre_{|w|-1}(w)) \leq i-3$.
    \end{proof}

    Next, we describe how to use the algorithm in Figures \ref{algorithm:broadcast_K_values_source} and \ref{algorithm:broadcast_K_values_non_source} to broadcast a value from $\{1,\ldots, m\}$.
    Let $r$ be the smallest value such that $w_r\geq m$.
    %so $w_{r-1} < m$.
    Since $m \geq 2$, it follows that $r \geq 3$.
    % In Section \ref{section:appendix_upper_bound}, we show that $r \leq 2 \lceil \log_2{m}\rceil+2$.
    Let $S= W_r=\{s_1,s_2,\ldots,s_{w_r}\}$.
    To broadcast 
    $\mu \in \{1,\ldots, m\}$, the source node 
    sets $s = \mathcal{M}(\mu)$ and performs the algorithm
    in Figure \ref{algorithm:broadcast_K_values_source}.
   % For all $i \in \{1,\ldots, m\}$, let $\mathcal{M}(i)=s_i$.
    %Let $\mu \in \{1,\ldots, m\}$ be the value that the source node wishes to broadcast.
   % The only modification to source node's algorithm is that it broadcasts $\mathcal{M}(\mu)$ instead of $\mu$.
   Each non-source node performs the algorithm in Figure \ref{algorithm:broadcast_K_values_non_source}, but when it terminates
   it decodes the message
   as $\mathcal{M}^{-1}(\textit{acc}_{v})$.
    % immediately before terminating,
    % each non-source node $v$ learns that the message is $\mathcal{M}^{-1}(\textit{acc}_{v})$.
    %where $\textit{acc}_{v}$ is the value of $\textit{acc}_v$.
    By observation \ref{observation:k_values:algorithm:source_node_beep_rounds} and Lemma \ref{lemma:k_values:algorithm:termination},
    all nodes terminate within $D+C(pre_{|\mathcal{M}(\mu)|-1]}(\mathcal{M}(\mu)))+3$ rounds. 
    Since $r \geq 3$, 
    Lemma \ref{lemma:k_values:algorithm:encoding_cost_upper_bound} implies that 
    $C(pre_{|\mathcal{M}(\mu)|-1]}(\mathcal{M}(\mu))) \leq r-3$. 
    Thus, all nodes learn $\mu$ and terminate within $D+r$ rounds. 
    
    It is known that
    %Note that 
    %According to Benoit Cloitre \cite{oeis}, 
    $w_{i}=\left\lfloor d c^{i}+\frac{1}{2}\right\rfloor$, where $c\approx 1.4656$ is the real root of $x^3-x^2-1$ and $d \approx 0.6115$ is the real root of $31\cdot x^3-31\cdot x^2+9 \cdot x-1$~\cite{oeis}. 
    Thus $\left\lfloor d c^{r-1}\right\rfloor = w_{r-1} < m$.
    Since $m$ is an integer, $d c^{r-1}< m$,
    so $r < \log_c{m}-\log_c{d}+1 < 2\log_{2}{m}+3$. 
    Since $r$ is an integer, $r \leq 2 \lceil \log_2{m}\rceil+2$.

    \begin{theorem} \label{theorem:main_upper_bound}
        Let $m\geq2$ and 
        let $r\geq 3$ be the smallest value such that $w_r \geq m$. 
        Consider the algorithm in Figures \ref{algorithm:broadcast_K_values_source} and \ref{algorithm:broadcast_K_values_non_source} that enables the source node to broadcast any message from $\{1,\ldots, m\}$.
        For all $D\geq1$ and all $\mu \in \{1,\ldots,m\}$,
        all nodes learn $\mu$ and terminate within $D+r \leq D+2 \lceil \log_2{m}\rceil+2$ rounds
        during the execution of the algorithm with message $\mu$ on every graph of source eccentricity $D$.
    \end{theorem}

% Add optimizations:    
    % First modification is for the source node. Instead of executing the loop for all characters of 
    % - Go from 1 to location of last 1 in $s$ when looping in the source node

\section{Lower Bound}    \label{section:lower_bound}
    It is convenient for the proof of the lower bound to strengthen the model 
    by assuming
    that each node knows its level and whether that level is the final level of the graph.
    We also assume that the source node is fixed. 
    Since a node that listens can simulate idling by throwing away any information it receives, we assume that nodes either listen or transmit at each round.
    
    We define a family of graphs used in the proof.
    For all $D \geq 2$, let $E_D$ be the graph with a source node and two nodes at each level $1$ through $D$, 
    such that, for all $0 \leq \ell \leq D-1$, each node at level $\ell$ is connected to both nodes at level $\ell+1$.
    Figure \ref{fig:E_3} shows graph $E_{3}$. 
     
    \begin{figure}[H]
        \centering
        \includegraphics[width=5cm]{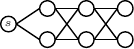}
        \caption{The graph $E_{3}$ with source node $s$.}
        \label{fig:E_3} 
    \end{figure}

    The main result in this section is the following theorem, where $w_0,w_1,w_2,\ldots$ is Narayana's cows sequence.
    It proves that the algorithm in Figures \ref{algorithm:broadcast_K_values_source} and \ref{algorithm:broadcast_K_values_non_source} is optimal in the radio broadcast model with collision detection.
    \begin{theorem} \label{theorem:main_lower_bound}
        Let $m \geq 2$ and let $r\geq 3$ be the smallest value such that $w_r \geq m$. 
        Consider any algorithm $\mathcal{B}$ that enables the source node to broadcast any message from $\{1,\ldots, m\}$.
        Then there exist $D \geq 12$, a message $\mu \in \{1,\ldots, m\}$,
        and a non-source node that uses at least $D+r$ rounds
        during the execution of $\mathcal{B}$ with message $\mu$ on $E_D$.
    \end{theorem}

    Let $m \geq 2$ and let $r \geq 3$ be the smallest value such that $w_r \geq m$. 
    Since $w_{r-1} < m \leq w_r$, an algorithm that enables the source node to broadcast any message from $\{1,\ldots, m\}$ also enables the source node to broadcast any message from $\{1,\ldots,w_{r-1}+1\}$. 
    Therefore, it suffices to assume that $m=w_{r-1}+1$.

    To obtain a contradiction,
    assume there exists an algorithm $\mathcal{B}$ that enables the source node to broadcast any message from $\{1,\ldots, w_{r-1}+1\}$, such that, for all $D \geq 12$ and all $\mu \in \{1,\ldots, w_{r-1}+1\}$,
    all non-source nodes terminate within $D+r-1$ rounds
    during the execution of $\mathcal{B}$ with message $\mu$ on $E_D$.
    
    The high-level road-map for the lower bound proof is as follows.
    % First, in section \ref{subsection:properties_of_general_algorithm}, we prove a number of useful properties of the execution of an arbitrary algorithm on $E_D$.
    First, in section \ref{subsection:properties_of_general_algorithm},
    we consider an arbitrary algorithm $\mathcal{A}$ such that, for all $D \geq 12$ and all messages $\mu \in M$, 
    all nodes 
    eventually wake up during the execution of $\mathcal{A}$ with message $\mu$ on $E_{D}$.
    We prove that, all nodes at the same level wake up at the same round during all such executions.
    Then we show that, for each level $\ell \geq 2$,
    when the nodes at level $\ell$ wake up,
    they have no information about the message or the graph provided the graph is sufficiently large.
    We also prove a number of other useful properties of $\mathcal{A}$ on $E_D$.
    
    In section \ref{subsection:construction_of_B'}, we define property $\textit{P}(k)$, which says that, 
    % for all $2\leq \ell \leq D$, 
    nodes at level $\ell$ wake up at round $\ell$ and distinguish between the messages from $\{1,\ldots, w_{k-1}+1\}$ within $k-1$ rounds after waking up. From algorithm $\mathcal{B}$, we construct an algorithm $\mathcal{B}'$ that has property  $\textit{P}(r)$,
    using the results in section \ref{subsection:properties_of_general_algorithm}.

   % Using these properties, in section \ref{subsection:construction_of_B'}, 
   % we construct an algorithm $\mathcal{B}'$ from $\mathcal{B}$ such that the nodes at level $2 \leq \ell \leq D$ wake up at round $\ell$ and distinguish between the messages from $\{1,\ldots, w_{r-1}+1\}$ within $r-1$ rounds after waking up.
    % This property is called $\textit{P}(r)$.
    
    % We formalize 
    % $\textit{P}(r)$ 
    % % that enables the source node to broadcast any message from $\{1,\ldots, w_{r-1}+1\}$, such that, for all $D \geq 12$ and all $\mu \in \{1,\ldots, w_{r-1}+1\}$,
    % % all non-source nodes terminate within $D+r-1$ rounds
    % % during the execution of $\mathcal{B}'$ with message $\mu$ on $E_D$.
    
    % we define property $\textit{P}(k)$, which says that, for all $2\leq \ell \leq D$, nodes at level $\ell$ wake up at round $\ell$ and distinguish between the messages from $\{1,\ldots, w_{k-1}+1\}$ within $k-1$ rounds after waking up.
    
    In section \ref{subsection:Construction_of_Algorithm}, we show how to construct an algorithm that has property $\textit{P}(k')$ from an algorithm that has property $\textit{P}(k)$ (and some additional properties), for some $k' < k$.
    Finally, in section \ref{subsection:P(K)_is_impossible}, we inductively prove that, for all $k \geq 3$, no algorithm has property $\textit{P}(k)$.
    Since algorithm $\mathcal{B}'$ has property $\textit{P}(r)$, this contradicts the existence of algorithm $\mathcal{B}'$ and, hence, algorithm $\mathcal{B}$ does not exist.
    
    % In particular, all nodes at level $D$ distinguish between
    % we create another algorithm $\mathcal{B}'$ that has the same property as $\mathcal{B}$, but also ensures that nodes at level $\ell$ wake up at round $\ell$.
    % In particular, all nodes at level $D$ distinguish between
    % the 
    % source messages within $r-1$ rounds after waking up.
    % We prove, more generally, that for all $2\leq \ell \leq D$, all nodes at level $\ell$ distinguish between the source messages within $r-1$ rounds after waking up.

    % \begin{itemize}
    %     \item Proof by contradiction
    %     \item Assume there exists an algorithm $\mathcal{B}$ that enables the source node to broadcast any message from $\{1,\ldots, m\}$, such that, for all $D \geq 12$ and all $\mu \in \{1,\ldots, m\}$,
    %     all non-source nodes terminate within $D+r-1$ rounds
    %     during the execution of $\mathcal{B}$ with message $\mu$ on $E_D$.
    %     % breaks the lower bound.
    %     \item Create another algorithm $\mathcal{B}'$ that has the same property, but also ensures that nodes at level $\ell$ wake up at round $\ell$.
    %     In particular, all nodes at level $D$ distinguish between
    %     %any pair of
    %     the 
    %     source messages within $r-1$ rounds after waking up.
    %     We prove, more generally, that for all $2\leq \ell \leq D$, all nodes at level $\ell$ distinguish between the source messages within $r-1$ rounds after waking up.
    %     \item We define property $P(k)$ 
    %     \item We show inductively that, for all $k \geq 3$, no algorithm has property $P(k)$.
    % \end{itemize}    

\subsection{Behaviour of Algorithms on $E_{D}$}  \label{subsection:properties_of_general_algorithm}
    Consider an arbitrary algorithm $\mathcal{A}$ in which the source node is given a message $\mu \in M$.
    Let $\alpha_{\mu,D}$ denote the execution of $\mathcal{A}$ with message $\mu$ on $E_D$.  
    Assume that, for all $D \geq 12$ and all $\mu \in M$, all nodes 
    eventually wake up during $\alpha_{\mu,D}$.

    By the model, both nodes at each level $\ell \geq 1$ are initially in the same state.
    Since they
    have the same set of neighbours, they always remain in the same state as one another.
    Let $\mathbf{s}_{\mathcal{A}}(D,\mu,\ell,t)$ be the state of nodes at level $\ell$ at the end of round $t$ during $\alpha_{\mu,D}$.
    Thus, when a node at level $\ell$ transmits, the other node at level $\ell$ also transmits at that round and, hence, their neighbours receive a collision signal.
    Therefore, during $\alpha_{\mu,D}$,
    each node at level $\ell \geq 2$ receives either a collision signal or nothing when it is listening.
    
   During $\alpha_{\mu,D}$, the nodes at level $\ell \geq 1$ wake up at the same round, $t_{\mathcal{A}}(D,\mu,\ell)$. Then $\mathbf{s}_{\mathcal{A}}(D,\mu,\ell,t_{\mathcal{A}}(D,\mu,\ell))$ denotes the state of these nodes at the end of the round in which they wake up.
    
    % Since nodes at each level have the same set of neighbours as one another, they wake up at the same round during $\alpha_{\mu,D}$.
    % Let $t_{\mathcal{A}}(D,\mu,\ell)$ be the round in which the nodes at level $\ell \geq 1$ wake up during $\alpha_{\mu,D}$.
    % Note that $\mathbf{s}_{\mathcal{A}}(D,\mu,\ell,t_{\mathcal{A}}(D,\mu,\ell))$ denotes the state of the nodes at level $\ell$ at the end of the round in which they 
    % were woken up during $\alpha_{\mu,D}$.

    % In the next several results,
    Next,
    we prove properties about the round in which nodes wake up.
    The source node wakes up before any other node and each node listens up to and including the round in which it wakes up.
    Thus, we have the following observation.
    \begin{observation} \label{observation:must_receive_first_message_from_previous_level} 
        Let $1 \leq \ell \leq D$.
        If nodes at level $\ell$ do not wake up in the first $t$ rounds of $\alpha_{\mu,D}$, then nodes at levels greater than $\ell$ do not wake up in the first $t+1$ rounds of $\alpha_{\mu,D}$.
    \end{observation} 
    
    The next result follows directly from Observation \ref{observation:must_receive_first_message_from_previous_level}.
    \begin{observation} \label{observation:wake_up_by_nodes_at_previous_level}
        Let $2 \leq \ell \leq D$.
        If nodes at level $\ell$ wake up at round $t$ of $\alpha_{\mu,D}$, then nodes at level $\ell-1$ transmit at round $t$.
    \end{observation} 
        
    Now, we use Observation \ref{observation:must_receive_first_message_from_previous_level} inductively to improve the result for nodes at levels greater than $\ell+1$.
    \begin{lemma} \label{lemma:lower_bound_t+k}
        Let $1 \leq \ell \leq D$. 
        If nodes at level $\ell$ do not wake up in the first $t$ rounds of $\alpha_{\mu,D}$, then, 
        for all $1 \leq \ell' \leq D-\ell$, the nodes at level $\ell+\ell'$ do not wake up in the first $t+\ell'$ rounds of $\alpha_{\mu,D}$.
    \end{lemma}
    \begin{proof}
        Suppose nodes at level $\ell$ do not wake up in the first $t$ rounds of $\alpha_{\mu,D}$.
        Let $\ell' \geq 1$ and assume that the nodes at level $\ell+\ell'-1$ do not wake up in the first $t+\ell'-1$ rounds of $\alpha_{\mu,D}$.
        By Observation \ref{observation:must_receive_first_message_from_previous_level},
        nodes at level $\ell+\ell'$ do not wake up in the first $t+\ell'$ rounds of $\alpha_{\mu,D}$.
    \end{proof}

    Since the source node is the only node that can transmit at round $1$, nodes at level $2$ do not wake up at round $1$.
    Thus, by Lemma \ref{lemma:lower_bound_t+k}, we have the following result.
    \begin{corollary} \label{corollary:lower_bound_at_least_i_for_level_i}
        For all $1 \leq \ell \leq D$, nodes at level $\ell$ do not wake up in the first $\ell-1$ rounds of $\alpha_{\mu,D}$.
    \end{corollary} 

    For all $\ell \geq 2$, nodes at level $\ell$ wake up by receiving a collision signal.
    Furthermore, nodes at level $\ell$ are in the same initial state in $\alpha_{\mu,D}$, for every $\mu \in M$ and every $D \geq 12$, provided $\ell$ is not the last level of $E_D$ and, hence,
    they transition to the same state when they wake up.
    We call this state $\mathbf{c}_{\mathcal{A}}(\ell,0)$.
    
    Similarly, the nodes at level $D$ are in the same initial state during $\alpha_{\mu,D}$, for all $\mu \in M$.
    Thus, the state to which the nodes at the last level transition after waking up does not depend on the value of the message.
    \begin{lemma} \label{lemma:same_state_upon_waking_up_on_same_graph}
        For all $\mu_0, \mu_1 \in M$, $\mathbf{s}_{\mathcal{A}}(D,\mu_0,\ell,t_{\mathcal{A}}(D,\mu_0,\ell))=\mathbf{s}_{\mathcal{A}}(D,\mu_1,\ell,t_{\mathcal{A}}(D,\mu_1,\ell))$.
    \end{lemma}

    Now, suppose all non-source nodes 
    % (except possibly in the last level) 
    from some level on
    transmit immediately after waking up.
    This allows us to identify
    the rounds in which nodes at greater levels 
    wake up.
    \begin{lemma} \label{lemma:local_round_i_for_ell_is_ell+i}
        Suppose there exists a level $2 \leq \ell \leq D-1$ such that, for all $\ell' \geq \ell$, nodes in state
        $\mathbf{c}_{\mathcal{A}}(\ell',0)$ transmit.
        If nodes at level $\ell$ wake up at round $t$ of $\alpha_{\mu,D}$,
        then, for all $1 \leq \ell' \leq D-\ell$, the nodes at level $\ell+\ell'$ wake up at round $t+\ell'$ of $\alpha_{\mu,D}$.
    \end{lemma}
    \begin{proof}
        Suppose nodes at level $\ell$ wake up at round $t$ of $\alpha_{\mu,D}$.
        Let $\ell' \geq 1$ and assume that the nodes at level $\ell+\ell'-1$ wake up at round $t+\ell'-1$ of $\alpha_{\mu,D}$.
        Then, by Observation \ref{observation:must_receive_first_message_from_previous_level}, the nodes at level $\ell+\ell'$ do not wake up in the first $t+\ell'-1$ rounds of $\alpha_{\mu,D}$.
        The nodes at level $\ell+\ell'-1$ are in state $\mathbf{c}_{\mathcal{A}}(\ell+\ell'-1,0)$ at the end of round $t+\ell'-1$ of $\alpha_{\mu,D}$.
        Since nodes in state $\mathbf{c}_{\mathcal{A}}(\ell+\ell'-1,0)$ transmit, the nodes at level $\ell+\ell'-1$ transmit at round $t+\ell'$ of $\alpha_{\mu,D}$.
        Thus, the nodes at level $\ell+\ell'$ wake up at round $t+\ell'$ of $\alpha_{\mu,D}$.
    \end{proof}

    We highlight the special case in which every node (except possibly the nodes in the last level) transmits in the round after it wakes up.
    \begin{corollary} \label{corollary:local_round_i_for_ell_is_ell+i}
        During $\alpha_{\mu,D}$, suppose
        the source node transmits at round $1$, 
        nodes at level $1$ transmit at round $2$, and,
        for all $2 \leq \ell \leq D-1$, nodes in state $\mathbf{c}_{\mathcal{A}}(\ell,0)$ transmit.
        Then, for all $1 \leq \ell \leq D$, the nodes at level $\ell$ wake up at round $\ell$.
    \end{corollary}

    %%%%%%%%%%%%%%%%%%%%%%%%%%%%%%%%%%%%%%%%%%%%%%%%% %%%%%%%%%%%%%%%%%%%%%%%%%%%%%%%%%%%%%%%%%%%%%%% %%%%%%%%%%%%%%%%%%%%%%%%%%%%%%%%%%%%%%%%%%%%%%%
    %%%%%%%%%%%%%%%%%%%%%%%%%%%%%%%%%%%%%%%%%%%%%%%%% %%%%%%%%%%%%%%%%%%%%%%%%%%%%%%%%%%%%%%%%%%%%%%% %%%%%%%%%%%%%%%%%%%%%%%%%%%%%%%%%%%%%%%%%%%%%%% %%%%%%%%%%%%%%%%%%%%%%%%%%%%%%%%%%%%%%%%%%%%%%%%% %%%%%%%%%%%%%%%%%%%%%%%%%%%%%%%%%%%%%%%%%%%%%%% %%%%%%%%%%%%%%%%%%%%%%%%%%%%%%%%%%%%%%%%%%%%%%%

    Suppose the nodes at level $\ell$ wake up at round $\ell$.
    We show that, if the nodes at some level cannot distinguish between two messages 
    $t$ rounds after waking up, 
    then the nodes at greater levels cannot distinguish between those two messages $t$ rounds after waking up.
    \begin{lemma} \label{lemma:distinguishability} 
        Let $\mu_0,\mu_1 \in M$.
        Assume that, 
        for all $1 \leq \ell \leq D$, 
        the nodes at level $\ell$ wake up at round $\ell$ 
        of $\alpha_{\mu_0,D}$ and $\alpha_{\mu_1,D}$.
        Furthermore, suppose there exist
        a level $2 \leq \ell \leq D-1$ and $t\geq 0$
        such that,
        for all $0 \leq t' \leq t$,
        the nodes at level $\ell$ are in the same states
        at the end of round $\ell+t'$ of $\alpha_{\mu_0,D}$ and  $\alpha_{\mu_1,D}$.
        Then, the same is true for nodes at levels greater than $\ell$.
    \end{lemma} 
    \begin{proof}
We prove by induction on $0\leq t' \leq t$ and $\ell < \ell' \leq D$ that the nodes at level $\ell'$ are in the same state 
at the end of round $\ell'+t'$ 
of $\alpha_{\mu_0,D}$ and $\alpha_{\mu_1,D}$. 

For all $\ell+1 \leq \ell' \leq D-1$, the nodes at level $\ell'$ wake up at round $\ell'$ and, hence,
the nodes at level $\ell'$ are in state $\mathbf{c}_{\mathcal{A}}(\ell',0)$ at the end of round $\ell'$
of $\alpha_{\mu_0,D}$ and $\alpha_{\mu_1,D}$.
Furthermore, by Lemma \ref{lemma:same_state_upon_waking_up_on_same_graph}, the nodes at level $D$ are in state 
$\mathbf{s}_{\mathcal{A}}(D,\mu_0,D,D)$ at the end of round $D$
of $\alpha_{\mu_0,D}$ and $\alpha_{\mu_1,D}$.

For all $\ell+2 \leq \ell' \leq D$, since nodes at level $\ell'$ wake up at round $\ell'$, it follows from
Observation \ref{observation:wake_up_by_nodes_at_previous_level} that the nodes at level $\ell'-1$ transmit at round $\ell'$ of $\alpha_{\mu_0,D}$. 
Hence, for all $\ell+1 \leq \ell' \leq D-1$, nodes in state $\mathbf{c}_{\mathcal{A}}(\ell',0)$ transmit.
Note that, there can be only one state to which nodes in state $\mathbf{c}_{\mathcal{A}}(\ell',0)$ transition after transmitting.
Therefore,  for all $\ell+1 \leq \ell' \leq D-1$, nodes at level $\ell'$ are in the same state at the end of
round $\ell'+1$ 
of $\alpha_{\mu_0,D}$ 
and 
$\alpha_{\mu_1,D}$. 

Recall that, at the end of round $D$ of $\alpha_{\mu_1,D}$, the nodes at levels $D-1$ and $D$ are in the same states
as the nodes at levels $D-1$ and $D$ at the end of round $D$ of $\alpha_{\mu_0,D}$.
Thus, at the end of round $D+1$ of $\alpha_{\mu_1,D}$, the nodes at level $D$ transition to the same state as the nodes at level $D$ at the end of round $D+1$ of $\alpha_{\mu_0,D}$.
Therefore, for all $\ell+1 \leq \ell' \leq D$, the nodes at level $\ell'$ are in the same states at the end of rounds $\ell'$ and $\ell'+1$ 
of $\alpha_{\mu_0,D}$ and $\alpha_{\mu_1,D}$. 
Thus, the claim holds for $t'=0$ and $t'=1$.

    Let $t' \geq 2$ and assume the claim holds for $t'-1$ and $t'-2$, and nodes at all levels from $\ell$ to $D$.
    Also, let $\ell < \ell' \leq D$ and assume the claim holds for $t'$ and nodes at level $\ell'-1$.
    
        By the induction hypothesis, 
        the nodes at level $\ell'-1$ are in
        the same state
        at the end of round $\ell'+t'-1$ 
        of $\alpha_{\mu_0,D}$ and 
        $\alpha_{\mu_1,D}$. 
        Similarly, the nodes at level $\ell'$ are in
        the same state 
        at the end of round $\ell'+t'-1$ 
        of $\alpha_{\mu_0,D}$ and 
        $\alpha_{\mu_1,D}$. 
    
        First, suppose that $\ell' = D$.
        Since each node at levels $D-1$ and $D$ is in the same state
        at the end of round $D+t'-1$ 
        of $\alpha_{\mu_0,D}$ and 
        $\alpha_{\mu_1,D}$,
        the nodes at level $D$ are in
        the same state
        at the end of round $D+t'$ 
        of $\alpha_{\mu_0,D}$ and 
        $\alpha_{\mu_1,D}$. 
        Thus, the claim holds for the nodes at level $D$ at the end of round $t'$.
        
        Otherwise, $\ell' \leq D-1$. 
        By the induction hypothesis, the nodes at level $\ell'+1$ are in
        the same state
        at the end of round $(\ell'+1)+t'-2=\ell'+t'-1$ 
        of $\alpha_{\mu_0,D}$ and 
        $\alpha_{\mu_1,D}$. 
        Since each node at levels $\ell'-1$, $\ell'$ and $\ell'+1$ is in the same state
        at the end of round $\ell'+t'-1$ 
        of $\alpha_{\mu_0,D}$ and 
        $\alpha_{\mu_1,D}$,
        the nodes at level $\ell'$ are in
        the same state
        at the end of round $\ell'+t'$ 
        of $\alpha_{\mu_0,D}$ and 
        $\alpha_{\mu_1,D}$. 
        Thus, the claim holds for the nodes at level $\ell'$ at round $t'$.
    \end{proof}

    %%%%%%%%%%%%%%%%%%%%%%%%%%%%%%%%%%%%%%%%%%%%%%%%% %%%%%%%%%%%%%%%%%%%%%%%%%%%%%%%%%%%%%%%%%%%%%%% %%%%%%%%%%%%%%%%%%%%%%%%%%%%%%%%%%%%%%%%%%%%%%%
    %%%%%%%%%%%%%%%%%%%%%%%%%%%%%%%%%%%%%%%%%%%%%%%%% %%%%%%%%%%%%%%%%%%%%%%%%%%%%%%%%%%%%%%%%%%%%%%% %%%%%%%%%%%%%%%%%%%%%%%%%%%%%%%%%%%%%%%%%%%%%%% %%%%%%%%%%%%%%%%%%%%%%%%%%%%%%%%%%%%%%%%%%%%%%%%% %%%%%%%%%%%%%%%%%%%%%%%%%%%%%%%%%%%%%%%%%%%%%%% %%%%%%%%%%%%%%%%%%%%%%%%%%%%%%%%%%%%%%%%%%%%%%%
    
    Let $D \geq 12$ and
    assume that the nodes at the last $7$ levels of $E_D$ are in the same sequence of states at the end of some round of $\alpha_{\mu_0,D}$ as they are at the end of some (possibly different) round of $\alpha_{\mu_1,D}$.
    In the next two results, we show that the same is true 
    % identify the similarity between the states of nodes in the last $7$ levels 
    at the end of the next round of $\alpha_{\mu_0,D}$ and $\alpha_{\mu_1,D}$.
    First, we consider the case in which nodes at level $D-6$ transmit.
    \begin{observation} \label{observation:D-6_to_D_same_seq_of_states_and_D-6_transmit} 
        Suppose that the nodes at levels $D-6$ through $D$ are in the same states at the end of round $t$ of $\alpha_{\mu_0,D}$ as they are at the end of round $t'$ of $\alpha_{\mu_1,D}$, for some $t, t' \geq 1$.
        If the nodes at level $D-6$ transmit at round $t+1$ of $\alpha_{\mu_0,D}$, then
        the nodes at levels $D-6$ through $D$ are in the same states at the end of round $t+1$ of $\alpha_{\mu_0,D}$ as they are at the end of round $t'+1$ of $\alpha_{\mu_1,D}$.
    \end{observation}
    
    Now, consider the case in which nodes at level $D-6$ listen.
    \begin{observation} \label{observation:D-6_to_D_same_seq_of_states_and_D-7_same_op} 
        Suppose that the nodes at levels $D-6$ through $D$ are in the same states at the end of round $t$ of $\alpha_{\mu_0,D}$ as they are at the end of round $t'$ of $\alpha_{\mu_1,D}$, for some $t, t' \geq 1$.
        If the nodes at level $D-6$ listen at round $t+1$ of $\alpha_{\mu_0,D}$ and nodes at level $D-7$ perform the same operation at round $t+1$ of $\alpha_{\mu_0,D}$ as they do at round $t'+1$ of $\alpha_{\mu_1,D}$,
        then the nodes at levels $D-6$ through $D$ are in the same states at the end of round $t+1$ of $\alpha_{\mu_0,D}$ as they are at the end of round $t'+1$ of $\alpha_{\mu_1,D}$.
    \end{observation}

    Assume that the nodes at the first $\ell+1$ levels are in the same sequence of states at the end of some round of $\alpha_{\mu,D}$ as they are at the end of some (possibly different) round of $\alpha_{\mu,D'}$, for some $D' > D \geq 12$.
    In the next two results, we show that the same is true 
    at the end of the next round of $\alpha_{\mu,D}$ and $\alpha_{\mu,D'}$.
    First, we consider the case in which the nodes at level $\ell$ transmit.
    \begin{observation} \label{observation:same_seq_of_states_and_transmit} 
        Suppose that the nodes at levels $0$ through $\ell$ are in the same states at the end of round $t$ of $\alpha_{\mu,D}$ as they are at the end of round $t'$ of $\alpha_{\mu,D'}$, for some $t, t' \geq 1$.
        If nodes at level $\ell$ transmit at round $t+1$ of $\alpha_{\mu,D}$,
        then the nodes at levels $0$ through $\ell$ are in the same states at the end of round $t+1$ of $\alpha_{\mu,D}$ as they are at the end of round $t'+1$ of $\alpha_{\mu,D'}$.
    \end{observation}
    
    Now, we consider the case in which nodes at level $\ell$ listen.
    \begin{observation} \label{observation:same_seq_of_states_and_same_op} 
        Suppose that the nodes at levels $0$ through $\ell$ are in the same states at the end of round $t$ of $\alpha_{\mu,D}$ as they are at the end of round $t'$ of $\alpha_{\mu,D'}$, for some $t, t' \geq 1$.
        If nodes at level $\ell$ listen at round $t+1$ of $\alpha_{\mu,D}$ and nodes at level $\ell+1$ perform the same operation at round $t+1$ of $\alpha_{\mu,D}$ as they do at round $t'+1$ of $\alpha_{\mu,D'}$,
        then the nodes at levels $0$ through $\ell$ are in the same states at the end of round $t+1$ of $\alpha_{\mu,D}$ as they are at the end of round $t'+1$ of $\alpha_{\mu,D'}$.
    \end{observation}

    Suppose the nodes at some level do not wake up in the first $t$ rounds of $\alpha_{\mu,D}$.
    The next result shows that the nodes at levels $0$ through $D-1$ cannot distinguish between the executions of $\mathcal{A}$ with message $\mu$ on graphs $E_{D}$ and $E_{D'}$ at the end of round $t$.
    \begin{lemma} \label{lemma:one_can_show_by_induction}
        Suppose that nodes at level $2 \leq \ell \leq D$ do not wake up in the first $t$ rounds of $\alpha_{\mu,D}$. 
        Then, at the end of round $t$ of $\alpha_{\mu,D}$, 
        the nodes at levels $0$ through $\ell-1$ are in the same states as they are at the end of round $t$ of $\alpha_{\mu,D'}$, 
        and nodes at level $\ell$ do not wake up in the first $t$ rounds of $\alpha_{\mu,D'}$.
    \end{lemma}
    \begin{proof}
        We prove the claim by induction on the round number $t \geq 1$.
        
        By the model, every node (except possibly the nodes at the last level of $E_D$) is in the same initial state in $\alpha_{\mu,D}$ and $\alpha_{\mu,D'}$.
        In $\alpha_{\mu,D}$ and $\alpha_{\mu,D'}$, the source node transitions to the same state after it is given message $\mu$ and, hence, the source node performs the same operation at round $1$ of $\alpha_{\mu,D}$ and $\alpha_{\mu,D'}$.
        Thus, the nodes at levels $0$ and $1$ transition to the same state at the end of round $1$ of $\alpha_{\mu,D}$ as the nodes at levels $0$ and $1$ at the end of round $1$ of $\alpha_{\mu,D'}$.
        By Corollary \ref{corollary:lower_bound_at_least_i_for_level_i}, the nodes at levels $2$ and greater receive nothing at round $1$ and, hence, they  are in the same states at the end of round $1$ of $\alpha_{\mu,D}$ and $\alpha_{\mu,D'}$.
        Therefore, 
        at the end of round $1$ of $\alpha_{\mu,D}$, 
        the nodes at levels $0$ through $\ell-1$ are in the same states as they are at the end of round $1$ of $\alpha_{\mu,D'}$, 
        and nodes at level $\ell$ do not wake up in the first round of $\alpha_{\mu,D'}$.
        Thus, the claim holds for $t=1$.
        
        Let $t\geq 2$ and assume the claim is true for $t-1$.
        Suppose nodes at level $\ell$ do not wake up in the first $t$ rounds of $\alpha_{\mu,D}$. 
        Then, nodes at level $\ell$ do not wake up in the first $t-1$ rounds of $\alpha_{\mu,D}$.
        By the induction hypothesis, 
        at the end of round $t-1$ of $\alpha_{\mu,D}$, 
        the nodes at levels $0$ through $\ell-1$ are in the same states as they are at the end of round $t-1$ of $\alpha_{\mu,D'}$, 
        and nodes at level $\ell$ do not wake up in the first $t-1$ rounds of $\alpha_{\mu,D'}$.
        Since nodes at level $\ell$ do not wake up in the first $t-1$ rounds of $\alpha_{\mu,D}$ and $\alpha_{\mu,D'}$, the nodes at level $\ell$ listen at round $t$ of $\alpha_{\mu,D}$ and $\alpha_{\mu,D'}$.
        By Observation \ref{observation:same_seq_of_states_and_same_op}, at the end of round $t$ of $\alpha_{\mu,D}$, 
        the nodes at levels $0$ through $\ell-1$ are in the same states as they are at the end of round $t$ of $\alpha_{\mu,D'}$, 
        
        By Observation \ref{observation:must_receive_first_message_from_previous_level}, nodes at level greater than $\ell$ do not wake up in the first $t-1$ rounds of $\alpha_{\mu,D'}$ and, hence, they listen at round $t$ of $\alpha_{\mu,D'}$. 
        By assumption, the nodes at level $\ell$ do not wake up in the first $t$ rounds of $\alpha_{\mu,D}$.
        Thus, the nodes at level $\ell-1$ do not transmit at round $t$ of $\alpha_{\mu,D}$. 
        Since the nodes at level $\ell-1$ are in the same state at the end of round $t-1$ of $\alpha_{\mu,D}$ and $\alpha_{\mu,D'}$, the nodes at level $\ell-1$ do not transmit at round $t$ of $\alpha_{\mu,D'}$. 
        Therefore, nodes at levels $\ell$ do not wake up in the first $t$ rounds of $\alpha_{\mu,D'}$.
    \end{proof}
    
    Now, we show that the nodes at levels $2$ through $D$ wake up at the same round during the executions of $\mathcal{A}$ with the same message on graphs $E_{D}$ and $E_{D'}$.
    \begin{lemma} \label{lemma:wake_up_at_same_round_on_E_D0_and_E_D1}
        Suppose nodes at level $2 \leq \ell \leq D$ wake up at round $t$ of $\alpha_{\mu,D}$. 
        Then, the nodes at level $\ell$ wake up at round $t$ of $\alpha_{\mu,D'}$. 
    \end{lemma}
    \begin{proof}
        By assumption, the nodes at level $\ell$ do not wake up in the first $t-1$ rounds of $\alpha_{\mu,D}$.
        By Observation \ref{observation:wake_up_by_nodes_at_previous_level}, nodes at level $\ell-1$ transmit at round $t$ of $\alpha_{\mu,D}$.
        By Lemma \ref{lemma:one_can_show_by_induction}, the nodes at level $\ell-1$ are in the same state at the end of round $t-1$ of $\alpha_{\mu,D}$ and $\alpha_{\mu,D'}$, and the nodes at level $\ell$ do not wake up in the first $t-1$ rounds of $\alpha_{\mu,D'}$.
        Since nodes at level $\ell-1$ transmit at round $t$ of $\alpha_{\mu,D}$, nodes at level $\ell-1$ transmit at round $t$ of $\alpha_{\mu,D'}$.
        Thus, the nodes at level $\ell$ wake up at round $t$ of $\alpha_{\mu,D'}$. 
    \end{proof}

%%%%%%%%%%%%%%%%%%%%%%%%%%%%%%%%%%%%%%%%%%%%%%%%%%%%%%%%%%%%%%%%%%%%%%%%%%%%%%%%%%%%%%%%%%%%
%%%%%%%%%%%%%%%%%%%%%%%%%%%%%%%%%%%%%%%%%%%%%%%%%%%%%%%%%%%%%%%%%%%%%%%%%%%%%%%%%%%%%%%%%%%%
%%%%%%%%%%%%%%%%%%%%%%%%%%%%%%%%%%%%%%%%%%%%%%%%%%%%%%%%%%%%%%%%%%%%%%%%%%%%%%%%%%%%%%%%%%%%

\subsection{Construction of $\mathcal{B}'$}  \label{subsection:construction_of_B'} 
    % It is more convenient for the lower bound proof
    % to consider a special class of algorithms, in particular, we would like to focus on algorithms that have property $\textit{P}(r)$.
    
    % It is convenient for the proof of the lower bound that the nodes at each level wake up as early as possible.
    % This reduces the number of cases that have to be considered.
    % In particular, nodes at level $D$ wake up at round $D$.
    % If the algorithm terminates within $D+r-1$ rounds, these nodes must distinguish between
    % all $m= w_{r-1} +1$  source messages within $r-1$ rounds after waking up.
    % We show that, for all $2\leq \ell \leq D-1$, all nodes at level $\ell$ also distinguish between all source messages within $r-1$ rounds after waking up.
    
    % It is convenient for the proof of the lower bound to fix the round in which the nodes at each level wake up. 
    % This reduces the number of cases that have to be considered.
    
    % Nodes at level $D$ wake up at round $D$, hence, they distinguish between
    % the source messages within $r-1$ rounds after waking up.
    % More generally, for all $2\leq \ell \leq D$, all nodes at level $\ell$ distinguish between the source messages within $r-1$ rounds after waking up.
    
    In this section, we construct an algorithm $\mathcal{B}'$ from $\mathcal{B}$,  which still terminates within $D+r-1$ rounds, but the nodes at each level wake up as early as possible. 
    This reduces the number of cases that will have to be considered.
    In particular, nodes at level $D$ wake up at round $D$ of 
    $\mathcal{B}'$ and, hence,
    % If the algorithm terminates within $D+r-1$ rounds,
    they distinguish between
    all $m= w_{r-1} +1$  messages within $r-1$ rounds after waking up.
    We show that, for all $2\leq \ell \leq D-1$, 
    during $\mathcal{B}'$,
    all nodes at level $\ell$ also distinguish between all source messages within $r-1$ rounds after waking up.
    These properties are formalised in the following definition.
    % 
    % Thus, algorithm $\mathcal{B}'$ has property $\textit{P}(r)$, where property $\textit{P}(k)$ is defined as follows.
    % We define a ter
    % 
    % First, we formally define property $\textit{P}(k)$.  
    \begin{definition}
        An algorithm {$\mathcal{A}$ has property $\textit{P}(k)$} if, for all $D \geq 12$, 
        \begin{enumerate}
            \item for all $\mu \in \{1,\ldots, w_{k-1}+1\}$ and for all $1 \leq \ell \leq D$,
                  the nodes at level $\ell$ wake up at round $\ell$ 
                  during the execution of $\mathcal{A}$ with message $\mu$ on $E_D$ and,
            \item for all distinct $\mu,\mu' \in \{1,\ldots,w_{k-1}+1\}$ and for all $2 \leq \ell \leq D-1$, 
                  there exists $0 \leq t_{\ell} \leq k-1$ such that
                  the nodes at level $\ell$ are in different states at the end of round $\ell+t_{\ell}$ during the executions of $\mathcal{A}$ with messages $\mu$ and $\mu'$ on $E_{D}$.
                %   the state of the nodes at level $\ell$ at the end of round $\ell+t_{\ell}$ 
                %   of the execution of $\mathcal{A}$ with message $\mu$ on $E_D$
                %   is different than 
                %   the state of the nodes at level $\ell$ at the end of round $\ell+t_{\ell}$ 
                %   of the execution of $\mathcal{A}$ with message $\mu'$ on $E_D$.
        \end{enumerate} 
    \end{definition}
    
    We will show that algorithm $\mathcal{B}'$ has property $\textit{P}(r)$.
    % In this section, we will create algorithm $\mathcal{B}'$, from algorithm $\mathcal{B}$, with property $\textit{P}(r)$.
    % 
    First, we show that, in $\mathcal{B}$, there are finitely many levels $\ell \geq 2$ in which nodes listen immediately after waking up.
    % First, we show that there are finitely many levels $\ell \geq 2$ such that nodes in state $\mathbf{c}_{\mathcal{B}}(\ell,0)$ listen.
    \begin{lemma} 
        Let $2 \leq \ell_{1} < \ell_2 < \cdots < \ell_{k}$ be an ordered list of levels such that,
        for all $\ell \in \{\ell_{1},\ell_2,\ldots,\ell_{k}\}$, nodes in state $\mathbf{c}_{\mathcal{B}}(\ell,0)$ listen.
        Then, $k < r$
    \end{lemma}
    \begin{proof}
        Let $D' = \ell_{k}+12$.
        We will show by induction on $1 \leq i \leq k$, that nodes at level $\ell_{i}+1$ do not wake up in the first $\ell_{i}+i$ rounds during the execution of $\mathcal{B}$ on $E_{D'}$.
        This implies that the nodes at level $\ell_{k}+1$ do not wake up in the first $\ell_{k}+k$ rounds.
        By Lemma \ref{lemma:lower_bound_t+k}, the nodes at level $D'=(\ell_{k}+1)+11$ do not wake up in the first $(\ell_{k}+k)+11=D'+k-1$ rounds.
        Thus, $k<r$ because we assumed that all nodes terminate within $D'+r-1$ rounds during the execution of $\mathcal{B}$ on $E_{D'}$.
        
        By Corollary \ref{corollary:lower_bound_at_least_i_for_level_i}, the nodes at level $\ell_{1}$ do not wake up in the first $\ell_{1}-1$ rounds.
        Similarly, the nodes at level $\ell_{1}+1$ do not wake up in the first $\ell_{1}$ rounds.
        Since nodes in state $\mathbf{c}_{\mathcal{B}}(\ell_1,0)$ listen, nodes at level $\ell_{1}$ listen after waking up.
        Hence, nodes at level $\ell_{1}$ listen at round $\ell_{1}+1$.
        Thus, the nodes at level $\ell_{1}+1$ do not wake up in the first $\ell_{1}+1$ rounds.
        Therefore, the claim holds for $i=1$.
        
        Let $i\geq 2$ and assume that the nodes at level $\ell_{i-1}+1$ do not wake up in the first $\ell_{i-1}+i-1$ rounds.
        Then, by Lemma \ref{lemma:lower_bound_t+k}, nodes at level $\ell_{i}$ do not wake up in the first $\ell_{i-1}+i-1+(\ell_{i}-\ell_{i-1}-1)=\ell_{i}+i-2$ rounds.
        Similarly, the nodes at level $\ell_{i}+1$ do not wake up in the first $\ell_{i}+i-1$ rounds.
        Since nodes in state $\mathbf{c}_{\mathcal{B}}(\ell_i,0)$ listen, nodes at level $\ell_{i}$ listen after waking up.
        Hence, nodes at level $\ell_{i}$ listen at round $\ell_{i}+i$.
        Thus, the nodes at level $\ell_{i}+1$ do not wake up in the first $\ell_{i}+i$ rounds.
    \end{proof}

    Let $\ell_0 \geq 11$ be such that nodes in state $\mathbf{c}_{\mathcal{B}}(\ell,0)$ transmit 
    for all $\ell \geq \ell_0$. 
    We construct algorithm $\mathcal{B}'$ 
    so that,
    for all $D \geq 12$,
    every
    execution of $\mathcal{B}'$ with a message $\mu \in \{1,\ldots, w_{r-1}+1\}$
    on $E_D$
    simulates 
    the execution of $\mathcal{B}$ with the same message
    on $E_{\ell_0+D}$:
    During the execution of $\mathcal{B}'$ 
    on $E_D$, 
    the source node simulates levels $0$ through $\ell_0$ in $E_{\ell_0+D}$ and
    each non-source node
    %, from when it wakes up,
    simulates a node $\ell_0$ levels higher in $E_{\ell_0+D}$.
    %, from when it wakes up.
    Nodes at level $\ell$ wake up at round $\ell$ and start simulating the nodes at level $\ell_0+\ell$ 
    of $E_{\ell_0+D}$ from the round in which they wake up
    in $\mathcal{B}$.

    % Now, we define a set of states that may be used by the source node during the execution of algorithms $\mathcal{B}'$ with message $\mu$.
    Now, we define the set of states $S_{\mathcal{B}}(\ell_0,\mu)$
    for the source node during the execution of algorithm $\mathcal{B}'$ with message $\mu$.
    % This set of states is defined for arbitrary algorithm $\mathcal{A}$
    We will also use this set of states for the algorithm constructed in section \ref{subsection:Construction_of_Algorithm}. Hence, we construct this set of states from an arbitrary algorithm $\mathcal{A}$ instead of from $\mathcal{B}$ and an arbitrary level $\ell' \geq 2$ instead of $\ell_0$.
    % For algorithm $\mathcal{A}$ and $\ell' \geq 2$, 
    Let 
    \begin{equation*}
        S_{\mathcal{A}}(\ell',\mu)=\{ [\mathbf{s}_{\mathcal{A}}(D,\mu,0,t),\ldots, \mathbf{s}_{\mathcal{A}}(D,\mu,\ell',t)] \mid \text{for all $D \geq \max\{12,\ell'+1\}$ and all $t \geq 1$} \}
    \end{equation*}
    consist of all sequences of states of nodes at levels $0$ through $\ell'$ that can occur at the end of some round during the execution of $\mathcal{A}$ with message $\mu$ on $E_D$, where $D$ is sufficiently large.
    Note that $S_{\mathcal{A}}(\ell',\mu)$ is disjoint from $S_{\mathcal{A}}(\ell',\mu')$, for $\mu \neq \mu'$, since the source node always knows the message.

    % \begin{definition}  \label{S_A(ell,mu)_operation}
        For any state $[\mathbf{s}_{\mathcal{A}}(D,\mu,0,t), \ldots, \mathbf{s}_{\mathcal{A}}(D,\mu,\ell',t)] \in S_{\mathcal{A}}(\ell',\mu)$, the source node performs the same operation as nodes in state $\mathbf{s}_{\mathcal{A}}(D,\mu,\ell',t)$ of $\mathcal{A}$.  
    % \end{definition}
    
    % \begin{definition} \label{S_A(ell,mu)_transmit}
        If the source node transmits in state $[\mathbf{s}_{\mathcal{A}}(D,\mu,0,t), \ldots, \mathbf{s}_{\mathcal{A}}(D,\mu,\ell',t)] \in S_{\mathcal{A}}(\ell',\mu)$, then it transitions to state $[\mathbf{s}_{\mathcal{A}}(D,\mu,0,t+1), \ldots, \mathbf{s}_{\mathcal{A}}(D,\mu,\ell',t+1)]$.
    % \end{definition}
    This transition is well defined by Observation \ref{observation:same_seq_of_states_and_transmit}.
    % By Observation \ref{observation:same_seq_of_states_and_transmit}, in any execution of $\mathcal{A}$ with message $\mu$ on $E_{D'}$ in which the nodes at levels $0$ through $\ell'$ are in
    % states $\mathbf{s}_{\mathcal{A}}(D,\mu,0,t), \ldots, \mathbf{s}_{\mathcal{A}}(D,\mu,\ell',t)$, respectively, 
    % at the end of some round $t'$
    % (i.e. $\mathbf{s}_{\mathcal{A}}(D,\mu,0,t')=\mathbf{s}_{\mathcal{A}}(D,\mu,0,t), \ldots, \mathbf{s}_{\mathcal{A}}(D,\mu,\ell',t')=\mathbf{s}_{\mathcal{A}}(D,\mu,\ell',t)$), 
    % the nodes at levels $0$ through $\ell'$ transition to states $\mathbf{s}_{\mathcal{A}}(D,\mu,0,t+1), \ldots,  \mathbf{s}_{\mathcal{A}}(D,\mu,\ell',t+1)$, respectively, (i.e. $\mathbf{s}_{\mathcal{A}}(D',\mu,0,t'+1)=\mathbf{s}_{\mathcal{A}}(D,\mu,0,t+1), \ldots, \mathbf{s}_{\mathcal{A}}(D',\mu,\ell',t'+1)=\mathbf{s}_{\mathcal{A}}(D,\mu,\ell',t+1)$), 
    % at the end of round $t'+1$.
    % Thus, this transition is well defined.
    
    Now, assume the source node in state $[\mathbf{s}_{\mathcal{A}}(D,\mu,0,t), \ldots, \mathbf{s}_{\mathcal{A}}(D,\mu,\ell',t)]$ listens.
    We separately define the transitions from this state
    for the cases when the source node receives a collision signal or nothing. 
    % \begin{definition} \label{S_A(ell,mu)_receive_col_signal}
        Suppose that the source node receives a collision signal while listening in state $[\mathbf{s}_{\mathcal{A}}(D,\mu,0,t), \ldots, \mathbf{s}_{\mathcal{A}}(D,\mu,\ell',t)]$.
        If nodes at level $\ell'+1$ transmit at round $t+1$ during the execution of $\mathcal{A}$ with message $\mu$ on $E_{D}$, then the source node in state $[\mathbf{s}_{\mathcal{A}}(D,\mu,0,t), \ldots, \mathbf{s}_{\mathcal{A}}(D,\mu,\ell',t)]$ transitions to state $[\mathbf{s}_{\mathcal{A}}(D,\mu,0,t+1), \ldots, \mathbf{s}_{\mathcal{A}}(D,\mu,\ell',t+1)]$.
        If nodes at level $\ell'+1$ listen at round $t+1$ during the execution of $\mathcal{A}$ with message $\mu$ on $E_{D}$, but there exist $t'$ and $D'\geq \max\{\ell'+1,12\}$
        such that
        the nodes at levels $0$ through $\ell'$ are in states $\mathbf{s}_{\mathcal{A}}(D,\mu,0,t), \ldots, \mathbf{s}_{\mathcal{A}}(D,\mu,\ell',t)$, respectively,
        at the end of round $t'$ during the execution of $\mathcal{A}$ with message $\mu$ on $E_{D'}$ and
        the nodes at level $\ell'+1$ transmit at round $t'+1$, then the source node in state $[\mathbf{s}_{\mathcal{A}}(D,\mu,0,t), \ldots, \mathbf{s}_{\mathcal{A}}(D,\mu,\ell',t)]$ transitions to state $[\mathbf{s}_{\mathcal{A}}(D',\mu,0,t'+1), \ldots, \mathbf{s}_{\mathcal{A}}(D',\mu,\ell',t'+1)]$.
        By Observation \ref{observation:same_seq_of_states_and_same_op}, this transition is well defined.
    If no such $t'$ and $D'$ exist, then, during the construction of the algorithm $\mathcal{A}'$, we make sure that the source node does not receive a collision signal while listening in state $[\mathbf{s}_{\mathcal{A}}(D,\mu,0,t), \ldots, \mathbf{s}_{\mathcal{A}}(D,\mu,\ell',t)]$.
    
    % \begin{definition} \label{S_A(ell,mu)_receive_nothing}
        The last case is when the source node receives nothing while listening in state  $[\mathbf{s}_{\mathcal{A}}(D,\mu,0,t), \ldots, \break \mathbf{s}_{\mathcal{A}}(D,\mu,\ell',t)]$.
        If nodes at level $\ell'+1$ listen at round $t+1$ during the execution of $\mathcal{A}$ with message $\mu$ on $E_{D}$, then the source node in state $[\mathbf{s}_{\mathcal{A}}(D,\mu,0,t), \ldots, \mathbf{s}_{\mathcal{A}}(D,\mu,\ell',t)]$ transitions to state $[\mathbf{s}_{\mathcal{A}}(D,\mu,0,t+1), \ldots, \mathbf{s}_{\mathcal{A}}(D,\mu,\ell',t+1)]$.
        Note that, in this case, the nodes at level $\ell'$ might receive a collision signal at round $t+1$ during the execution of $\mathcal{A}$ with message $\mu$ on $E_{D}$ from the nodes at level $\ell'-1$.
        If nodes at level $\ell'+1$ transmit at round $t+1$ during the execution of $\mathcal{A}$ with message $\mu$ on $E_{D}$, but there exist $t'$ and $D'\geq \max\{\ell'+1,12\}$
        such that
        the nodes at levels $0$ through $\ell'$ are in states $\mathbf{s}_{\mathcal{A}}(D,\mu,0,t), \ldots, \mathbf{s}_{\mathcal{A}}(D,\mu,\ell',t)$, respectively,
        at the end of round $t'$ during the execution of $\mathcal{A}$ with message $\mu$ on $E_{D'}$ and
        the nodes at level $\ell'+1$ listen at round $t'+1$, then the source node in state $[\mathbf{s}_{\mathcal{A}}(D,\mu,0,t), \ldots, \mathbf{s}_{\mathcal{A}}(D,\mu,\ell',t)]$ transitions to state $[\mathbf{s}_{\mathcal{A}}(D',\mu,0,t'+1), \ldots, \mathbf{s}_{\mathcal{A}}(D',\mu,\ell',t'+1)]$.
        By Observation \ref{observation:same_seq_of_states_and_same_op}, this transition is well defined.
    If no such $t'$ and $D'$ exist, then, during the construction of the algorithm $\mathcal{A}'$, we make sure that the source node receives a collision signal while listening in state $[\mathbf{s}_{\mathcal{A}}(D,\mu,0,t), \ldots, \mathbf{s}_{\mathcal{A}}(D,\mu,\ell',t)]$.

    \medskip

    Now we define the algorithm $\mathcal{B}'$ of $E_D$.
    \begin{itemize}
        \item $S_{\mathcal{B}}(\ell_0,\mu)$ is the set of states of the source node during the execution of $\mathcal{B}'$ with message $\mu$.
        % and transitions from these states are given in Definitions 
        % \ref{S_A(ell,mu)_transmit}, \ref{S_A(ell,mu)_receive_col_signal} and \ref{S_A(ell,mu)_receive_nothing}.
        Upon receiving message $\mu$ to broadcast, the source node transitions to state 
        $[\mathbf{s}_{\mathcal{B}}(\ell_0+1,\mu,0,t_{\mathcal{B}}(\ell_0+1,\mu,\ell_0)),\ldots, \mathbf{s}_{\mathcal{B}}(\ell_0+1,\mu,\ell_0,t_{\mathcal{B}}(\ell_0+1,\mu,\ell_0))]$.
        
        \item 
        For all $1 \leq \ell \leq D-1$, 
        then nodes at level $\ell$ transition to state $\mathbf{c}_{\mathcal{B}}(\ell_0+\ell,0)$ after waking up.
        After waking up, the nodes at level $D$ transition to state $\mathbf{s}_{\mathcal{B}}(\ell_0+D,\mu,\ell_0+D,t_{\mathcal{B}}(\ell_0+D,\mu,\ell_0+D))$.
        % where $\mu' \in M$ is arbitrary.
        Recall that, by Lemma \ref{lemma:same_state_upon_waking_up_on_same_graph}, the state to which the nodes at the last level transition after waking up does not depend on $\mu$.
        
        \item At every round in which the nodes at level $1$ receive a message, they act as though they received a collision signal.
    \end{itemize}

    Next, we show that
    every
    execution of $\mathcal{B}'$
    on $E_D$
    simulates 
    the execution of $\mathcal{B}$ with the same message
    on $E_{\ell_0+D}$.
    \begin{lemma} \label{lemma:k_values:q_0_transmit}
        For all $D \geq 12$, all $\mu \in \{1,\ldots, w_{r-1}+1\}$ and all $t\geq 0$, 
        the source node is in state
        $[\mathbf{s}_{\mathcal{B}}(\ell_0+D,\mu,0,t_{\mathcal{B}}(\ell_0+D,\mu,\ell_0)+t+1), \ldots, \mathbf{s}_{\mathcal{B}}(\ell_0+D,\mu,\ell_0,t_{\mathcal{B}}(\ell_0+D,\mu,\ell_0)+t+1)]$
        at the end of round $t+1$ of the execution of $\mathcal{B}'$ with message $\mu$ on $E_{D}$.
        Furthermore, for all $1 \leq \ell \leq D$, 
        during the execution of $\mathcal{B}'$ with message $\mu$ on $E_{D}$,
        the nodes at level $\ell$ wake up at round $\ell$ and,  
        at the end of round $\ell+t$, 
        the nodes at level $\ell$ are in the same state
        as the nodes at level $\ell_0+\ell$
        at the end of round $t_{\mathcal{B}}(\ell_0+D,\mu,\ell_0+\ell)+t$ 
        of the execution of $\mathcal{B}$ with message $\mu$ on $E_{\ell_0+D}$.
    \end{lemma}
    \begin{proof}
        Let $D \geq 12$ and $\mu \in \{1,\ldots, w_{r-1}+1\}$.
        Let $\beta_{\mu,D}$ and $\beta_{\mu,D}'$ denote the executions of $\mathcal{B}$ and $\mathcal{B}'$ with message $\mu$ on $E_{D}$, respectively.
        We prove the claim by induction on $t \geq 0$.
        
        Recall that the nodes at level $\ell_0$ wake up at round $t_{\mathcal{B}}(\ell_0+1,\mu,\ell_0)$ of $\beta_{\mu,\ell_0+1}$.
        By Lemma \ref{lemma:wake_up_at_same_round_on_E_D0_and_E_D1}, the nodes at level $\ell_0$ wake up at round $t_{\mathcal{B}}(\ell_0+1,\mu,\ell_0)$ of $\beta_{\mu,\ell_0+D}$, i.e. $t_{\mathcal{B}}(\ell_0+D,\mu,\ell_0)=t_{\mathcal{B}}(\ell_0+1,\mu,\ell_0)$.
        Thus, by Observation \ref{observation:must_receive_first_message_from_previous_level}, the nodes at level $\ell_0+1$ do not wake up in the first $t_{\mathcal{B}}(\ell_0+1,\mu,\ell_0)$ rounds of $\beta_{\mu,\ell_0+1}$.
        Therefore, by Lemma \ref{lemma:one_can_show_by_induction} 
        (for graphs $E_{\ell_0+1}$ and $E_{\ell_0+D}$), at the end of round $t_{\mathcal{B}}(\ell_0+D,\mu,\ell_0)$ of $\beta_{\mu,\ell_0+1}$,
        the nodes at levels $0$ through $\ell_0$ are in the same states as they are at the end of round $t_{\mathcal{B}}(\ell_0+D,\mu,\ell_0)$ of $\beta_{\mu,\ell_0+D}$.
        % 
        % Since the nodes at level $\ell_0$ wake up at round $t_{\mathcal{B}}(\ell_0+1,\mu,\ell_0)$ of $\beta_{\mu,\ell_0+1}$,
        The nodes at level $\ell_0$ are in state $\mathbf{c}_{\mathcal{B}}(\ell_0,0)$ at the end of round $t_{\mathcal{B}}(\ell_0+D,\mu,\ell_0)$ of $\beta_{\mu,\ell_0+D}$.
        By assumption, 
        for all $\ell \geq \ell_0$, nodes in state 
        $\mathbf{c}_{\mathcal{B}}(\ell,0)$ transmit.
        % In particular, the nodes at level $\ell_0$ transmit and nodes at level $\ell_0+1$ wake up 
        % at round $t_{\mathcal{B}}(\ell_0+1,\mu,\ell_0)+1$ of $\beta_{\mu,\ell_0+D}$.
        Thus, by Lemma \ref{lemma:local_round_i_for_ell_is_ell+i}, 
        for all $1 \leq \ell \leq D$, the nodes at level $\ell_0+\ell$ wake up at round $t_{\mathcal{B}}(\ell_0+D,\mu,\ell_0)+\ell=t_{\mathcal{B}}(\ell_0+D,\mu,\ell_0+\ell)$ of $\beta_{\mu,\ell_0+D}$. 
        For $1 \leq \ell \leq D-1$, the nodes at level $\ell_0+\ell$ are in state $\mathbf{c}_{\mathcal{B}}(\ell_0+\ell,0)$ when they wake up.
        The nodes at level $\ell_0+D$ are in state $\mathbf{s}_{\mathcal{B}}(\ell_0+D,\mu,\ell_0+D,t_{\mathcal{B}}(\ell_0+D,\mu,\ell_0)+D)$ at the end of round $t_{\mathcal{B}}(\ell_0+D,\mu,\ell_0)+D$ of $\beta_{\mu,\ell_0+D}$.
        
        By construction of $\mathcal{B}'$, after receiving message $\mu$ to broadcast, the source node transitions to state
        $[\mathbf{s}_{\mathcal{B}}(\ell_0+1,\mu,0,t_{\mathcal{B}}(\ell_0+1,\mu,\ell_0)),\ldots, \mathbf{s}_{\mathcal{B}}(\ell_0+1,\mu,\ell_0,t_{\mathcal{B}}(\ell_0+1,\mu,\ell_0))]$.
        Since nodes at levels $0$ through $\ell_0$ are in the same state at the end of round $t_{\mathcal{B}}(\ell_0+D,\mu,\ell_0)$ of $\beta_{\mu,\ell_0+1}$ and $\beta_{\mu,\ell_0+D}$,
        this state is equal to state
        $[\mathbf{s}_{\mathcal{B}}(\ell_0+D,\mu,0,t_{\mathcal{B}}(\ell_0+D,\mu,\ell_0)),\ldots, \mathbf{s}_{\mathcal{B}}(\ell_0+D,\mu,\ell_0,t_{\mathcal{B}}(\ell_0+D,\mu,\ell_0))]$.
        Recall that, in any state in $S_{\mathcal{B}}(\ell_0,\mu)$, the source node performs the same operation as the nodes perform during $\mathcal{B}$ when their state is the last component of that state.
        Since $\mathbf{s}_{\mathcal{B}}(\ell_0+D,\mu,\ell_0,t_{\mathcal{B}}(\ell_0+D,\mu,\ell_0))=\mathbf{c}_{\mathcal{B}}(\ell_0,0)$ and nodes in state $\mathbf{c}_{\mathcal{B}}(\ell_0,0)$ transmit, the source node transmits in the next round which is round $1$ of $\beta_{\mu,D}'$. 
        By definition of the transitions of states in $S_{\mathcal{B}}(\ell_0,\mu)$, at the end of round $1$, the source node transitions to state 
        $[\mathbf{s}_{\mathcal{B}}(\ell_0+D,\mu,0,t_{\mathcal{B}}(\ell_0+D,\mu,\ell_0)+1),\ldots, \mathbf{s}_{\mathcal{B}}(\ell_0+D,\mu,\ell_0,t_{\mathcal{B}}(\ell_0+D,\mu,\ell_0)+1)]$.

        By construction of $\mathcal{B}'$,
        for all $1 \leq \ell \leq D-1$, the nodes at level $\ell$ transition to state $\mathbf{c}_{\mathcal{B}}(\ell_0+\ell,0)$ after waking up.
        By assumption, for all $\ell \geq 0$, nodes in state $\mathbf{c}_{\mathcal{B}}(\ell_0+\ell,0)$ transmit.
        Thus, by Corollary \ref{corollary:local_round_i_for_ell_is_ell+i}, for all $1 \leq \ell \leq D$, the nodes at level $\ell$ wake up at round $\ell$ of $\beta_{\mu,D}'$.
        
        By construction of $\mathcal{B}'$, the nodes at level $D$ are in state $\mathbf{s}_{\mathcal{B}}(\ell_0+D,\mu,\ell_0+D,t_{\mathcal{B}}(\ell_0+D,\mu,\ell_0)+D)$ when they wake up $\beta_{\mu,D}'$.
        Thus, the claim holds for all levels when $t=0$.
        
        \medskip
        
        At the end of round $1$ of $\beta_{\mu,D}'$, the source node is in state $[\mathbf{s}_{\mathcal{B}}(\ell_0+D,\mu,0,t_{\mathcal{B}}(\ell_0+D,\mu,\ell_0)+1),\ldots,   \mathbf{s}_{\mathcal{B}}(\ell_0+D,\mu,\ell_0,t_{\mathcal{B}}(\ell_0+D,\mu,\ell_0)+1)]$ and 
        the nodes at level $1$ are in state $\mathbf{c}_{\mathcal{B}}(\ell_0+1,0)$.
        At round $2$ of $\beta_{\mu,D}'$, the source node and the nodes at level $1$ perform the same operation as the nodes at level $\ell_0$ and $\ell_0+1$, respectively, 
        at round $t_{\mathcal{B}}(\ell_0+D,\mu,\ell_0)+2$ of $\beta_{\mu,\ell_0+D}$.
        By definition of the transitions of states in $S_{\mathcal{B}}(\ell_0,\mu)$, 
        at the end of round $2$ of $\beta_{\mu,D}'$, the source node is in state 
        $[\mathbf{s}_{\mathcal{B}}(\ell_0+D,\mu,0,t_{\mathcal{B}}(\ell_0+D,\mu,\ell_0)+2),\ldots, \mathbf{s}_{\mathcal{B}}(\ell_0+D,\mu,\ell_0,t_{\mathcal{B}}(\ell_0+D,\mu,\ell_0)+2)] \in S_{\mathcal{B}}(\ell_0,\mu)$.
        
        For all $1 \leq \ell \leq D-1$, nodes in state $\mathbf{c}_{\mathcal{B}}(\ell_0+\ell,0)$ transmit and, hence, they can transition to only one state.
        Therefore, for all $1 \leq \ell \leq D-1$, the nodes at level $\ell$ are in the same state at the end of
        round $\ell+1$ of $\beta_{\mu,D}'$
        as the nodes at level $\ell_0+\ell$ at the end of round $t_{\mathcal{B}}(\ell_0+D,\mu,\ell_0)+\ell+1$ of $\beta_{\mu,\ell_0+D}$.
        
        Recall that, at the end of round $D$ of $\beta_{\mu,D}'$, the nodes at levels $D-1$ and $D$ are in the same states as the nodes at levels $\ell_0+D-1$ and $\ell_0+D$ at the end of round $t_{\mathcal{B}}(\ell_0+D,\mu,\ell_0)+D$ of $\beta_{\mu,\ell_0+D}$.
        Thus, at the end of round $D+1$ of $\beta_{\mu,D}'$, the nodes at level $D$ are in the same state as the nodes at level $\ell_0+D$ at the end of round $t_{\mathcal{B}}(\ell_0+D,\mu,\ell_0)+D+1$ of $\beta_{\mu,\ell_0+D}$.
        Thus, the claim holds for all levels when $t=1$.
        
        \medskip
        
        Let $t \geq 2$ and assume the claim holds for $t-1$ and $t-2$.
        By the induction hypothesis, at the end of round $(t-1)+1$ of $\beta_{\mu,D}'$, the source node is in state
        $[\mathbf{s}_{\mathcal{B}}(\ell_0+D,\mu,0,t_{\mathcal{B}}(\ell_0+D,\mu,\ell_0)+t),\ldots, \mathbf{s}_{\mathcal{B}}(\ell_0+D,\mu,\ell_0,t_{\mathcal{B}}(\ell_0+D,\mu,\ell_0)+t)]$.
        Furthermore, at the end of round $t=1+(t-1)=2+(t-2)$ of $\beta_{\mu,D}'$, the nodes at levels $1$ and $2$ are in the same states as the nodes at levels $\ell_0+1$ and $\ell_0+2$ at the end of round $t_{\mathcal{B}}(\ell_0+D,\mu,\ell_0)+t$ of $\beta_{\mu,\ell_0+D}$, respectively. 
        At round $t+1$ of $\beta_{\mu,D}'$, the source node, the nodes at level $1$, and the nodes at level $2$ perform the same operation as the nodes at level $\ell_0$, $\ell_0+1$, and $\ell_0+2$, respectively, 
        at round $t_{\mathcal{B}}(\ell_0+D,\mu,\ell_0)+t+1$ of $\beta_{\mu,\ell_0+D}$.
        Hence, the nodes at level $1$ transition to the same state as the nodes at level $\ell_0+1$ at the end of round $t_{\mathcal{B}}(\ell_0+D,\mu,\ell_0)+t+1$ of $\beta_{\mu,\ell_0+D}$. 
        Moreover, 
        by definition of the transitions of states in
        $S_{\mathcal{B}}(\ell_0,\mu)$, 
        at the end of round $t+1$ of $\beta_{\mu,D}'$, 
        the source node transitions to the state 
        $[\mathbf{s}_{\mathcal{B}}(\ell_0+D,\mu,0,t_{\mathcal{B}}(\ell_0+D,\mu,\ell_0)+t+1),\ldots, \mathbf{s}_{\mathcal{B}}(\ell_0+D,\mu,\ell_0,t_{\mathcal{B}}(\ell_0+D,\mu,\ell_0)+t+1)] \in S_{\mathcal{B}}(\ell_0,\mu)$.
        Thus, the claim holds for the nodes at levels $0$ and $1$ at round $t$.
        
        % At round $t+1$ of $\beta_{\mu,D}'$, the source node performs the same operation as the nodes at level $\ell_0$
        % at round $t_{\mathcal{B}}(\ell_0+D,\mu,\ell_0)+t+1$ of $\beta_{\mu,\ell_0+D}$.
        % Additionally, at round $t+1$ of $\beta_{\mu,D}'$, the nodes at levels $1$ and $2$ perform the same operations as the nodes at levels $\ell_0+1$ and $\ell_0+2$ at round $t_{\mathcal{B}}(\ell_0+D,\mu,\ell_0)+t+1$ of $\beta_{\mu,\ell_0+D}$.
        % Note that, if the nodes at level $1$ receive a message at round $t+1$ of $\beta_{\mu,D}'$, the nodes at level $\ell_0+1$ receive a collision signal at round $t_{\mathcal{B}}(\ell_0+D,\mu,\ell_0)+t+1$ of $\beta_{\mu,\ell_0+D}$.
        % However, in this case, by construction of $\mathcal{B}'$, the nodes at level $1$ act as though they received a collision signal.
        % Thus, at the end of round $t+1$ of $\beta_{\mu,D}'$, 
        % the nodes at level $1$ transition to the same state as the nodes at level $\ell_0+1$ at the end of round $t_{\mathcal{B}}(\ell_0+D,\mu,\ell_0)+t+1$ of $\beta_{\mu,\ell_0+D}$. 
        
        Let $1 < \ell \leq D$. 
        Assume that the nodes at level $\ell-1$ are in the same states at the end 
        round $(\ell-1)+t$ of $\beta_{\mu,D}'$
        as the nodes at level $\ell_0+(\ell-1)$ at the end round $t_{\mathcal{B}}(\ell_0+D,\mu,\ell_0)+(\ell-1)+t$ of $\beta_{\mu,\ell_0+D}$.
        Since the claim holds for $t-1$, the nodes at level $\ell$ are in the same states at the end of 
        round $\ell+t-1$ of $\beta_{\mu,D}'$
        as the nodes at level $\ell_0+\ell$ at the end of round $t_{\mathcal{B}}(\ell_0+D,\mu,\ell_0)+\ell+t-1$ of $\beta_{\mu,\ell_0+D}$.
        
        First, suppose that $\ell = D$.
        Since nodes at levels $\ell-1$ and $\ell$
        perform the same operations 
        at round $\ell+t$ 
        of $\beta_{\mu,D}'$ 
        as the nodes at levels $\ell_0+\ell-1$ and $\ell_0+\ell$
        at round $t_{\mathcal{B}}(\ell_0+D,\mu,\ell_0)+\ell+t$ of $\beta_{\mu,\ell_0+D}$,
        the nodes at level $\ell$ are in
        the same state
        at the end of round $\ell+t$ 
        of $\beta_{\mu,D}'$ 
        as the nodes at level $\ell_0+\ell$
        at the end of round $t_{\mathcal{B}}(\ell_0+D,\mu,\ell_0)+\ell+t$ of $\beta_{\mu,\ell_0+D}$.
        
        Otherwise, $\ell \leq D-1$. 
        Since the claim holds for $t-2$, the nodes at level $\ell+1$ are in the same states at the end of
        round $(\ell+1)+(t-2)$ of $\beta_{\mu,D}'$
        as the nodes at level $\ell_0+\ell+1$ at the end of round $t_{\mathcal{B}}(\ell_0+D,\mu,\ell_0)+(\ell+1)+(t-2)$ of $\beta_{\mu,\ell_0+D}$.
        Since nodes at levels $\ell-1$, $\ell$ and $\ell+1$
        perform the same operations
        at round $\ell+t$ 
        of $\beta_{\mu,D}'$ 
        as the nodes at levels $\ell_0+\ell-1$, $\ell_0+\ell$ and $\ell_0+\ell+1$ 
        at round $t_{\mathcal{B}}(\ell_0+D,\mu,\ell_0)+\ell+t$ of $\beta_{\mu,\ell_0+D}$,
        the nodes at level $\ell$ are in
        the same state
        at the end of round $\ell+t$ 
        of $\beta_{\mu,D}'$ 
        as the nodes at level $\ell_0+\ell$
        at the end of round $t_{\mathcal{B}}(\ell_0+D,\mu,\ell_0)+\ell+t$ of $\beta_{\mu,\ell_0+D}$.
        Thus, the claim holds for all levels at round $t$.
    \end{proof}

    By Lemma \ref{lemma:k_values:q_0_transmit}, the nodes at level $\ell$ wake up at round $\ell$ during the execution of $\mathcal{B}'$. 
    Thus, $\mathcal{B}'$ satisfies the first condition of property $\textit{P}(r)$.
              
    % In the next two results, we show that $\mathcal{B}'$ also satisfies the second condition of property $\textit{P}(r)$.
    Like algorithm $\mathcal{B}$, we show that algorithm $\mathcal{B}'$ also terminates within $D+r-1$ during executions on $E_D$.
    \begin{lemma} \label{lemma:B:ell_wake_up_at_ell}
        Algorithm $\mathcal{B}'$ enables the source node to broadcast any message from $\{1,\ldots, w_{r-1}+1\}$, such that, 
        for all $D \geq 12$, all $1 \leq \ell \leq D$ and all $\mu \in \{1,\ldots, w_{r-1}+1\}$,
        all non-source nodes terminate by the end of round $D+r-1$ during the execution of $\mathcal{B}'$ with message $\mu$ on $E_D$.
    \end{lemma}
    \begin{proof}
        By Corollary \ref{corollary:lower_bound_at_least_i_for_level_i}, nodes at level $\ell_0+\ell$ wake up at round $\ell_0+\ell$ or after during the execution of $\mathcal{B}$ with message $\mu \in \{1,\ldots, w_{r-1}+1\}$ on $E_{\ell_0+D}$.
        Since all non-source nodes terminate within $\ell_0+D+r-1$
        rounds,
        the nodes at level $\ell_0+\ell$ terminate within
        $D-\ell+r-1$ rounds after waking up.
        
        By Lemma \ref{lemma:k_values:q_0_transmit}, for all $1 \leq \ell \leq D$, 
        during the execution of $\mathcal{B}'$ with message $\mu$ on $E_{D}$,
        the nodes at level $\ell$ wake up at round $\ell$.
        Furthermore, for all $t\geq 0$, at the end of round $\ell+t$, 
        the nodes at level $\ell$ are in the same state
        as the nodes at level $\ell_0+\ell$
        at the end of round $t_{\mathcal{B}}(\ell_0+D,\mu,\ell_0+\ell)+t$
        of the execution of $\mathcal{B}$ with message $\mu$ on $E_{\ell_0+D}$.
        Since the nodes at level $\ell_0+\ell$ terminate within $D-\ell+r-1$ rounds after waking up (at round $t_{\mathcal{B}}(\ell_0+D,\mu,\ell_0+\ell)$) during
        the execution of $\mathcal{B}$ with message $\mu$ on $E_{\ell_0+D}$, 
        the nodes at level $\ell$ terminate within $D-\ell+r-1$ rounds after waking up (at round $\ell$)
        during
        the execution of
        $\mathcal{B}'$ with message $\mu$ on $E_D$.
    \end{proof}
    
    % During the execution of $\mathcal{B}'$ with message $\mu \in \{1,\ldots, w_{r-1}+1\}$ on $E_D$, nodes at level $D$ wake up at round $D$.
    % Therefore, all nodes at level $D$ distinguish between
    % the source messages within $r-1$ rounds after waking up.
    % We now prove, more generally, that for all $2\leq \ell \leq D$, all nodes at level $\ell$ distinguish between the source messages within $r-1$ rounds after waking up.
    Now, we show that $\mathcal{B}'$ satisfies the second condition of property $\textit{P}(r)$. 
    In particular, for any pair of messages in $\{1,\ldots, w_{r-1}+1\}$, the nodes at level $2 \leq \ell \leq D-1$ must be able to distinguish between them within $r-1$ rounds after waking up.
    \begin{lemma} \label{lemma:B:Distinguish} 
        For all distinct $\mu,\mu' \in \{1,\ldots, w_{r-1}+1\}$, all $D \geq 12$ and all $2 \leq \ell \leq D-1$, there exists $0 \leq t_{\ell} \leq r-1$
        such that,
        the nodes at level $\ell$ are in different states at the end of round $\ell+t_{\ell}$ during the executions of $\mathcal{B}'$ with messages $\mu$ and $\mu'$ on $E_{D}$.
    \end{lemma}
    \begin{proof}
        To obtain a contradiction, suppose there exist distinct $\mu ,\mu' \in \{1,\ldots, w_{r-1}+1\}$, $D \geq 12$ and
        a level $2 \leq \ell \leq D-1$ 
        such that, 
        for all $0 \leq t \leq r-1$,
        the nodes at level $\ell$ are in the same states at the end of round $\ell+t$ of the execution with message $\mu$ on $E_D$
        as they are at the end of round $\ell+t$ of the execution with message $\mu'$ on $E_D$.
        By Lemma \ref{lemma:distinguishability},
        for all $0 \leq t \leq r-1$, 
        the nodes at level $D$ are in the same states 
        at the end of round $D+t$ of the execution of $\mathcal{B}'$ with message $\mu$ on $E_D$
        as they are 
        at the end of round $D+t$ of the execution of $\mathcal{B}'$ with message $\mu'$ on $E_D$.
        Hence, the nodes at level $D$ cannot distinguish between $\mu$ and $\mu'$ at the end of first $D+r-1$ rounds of the execution of $\mathcal{B}'$ with message $\mu$ on $E_D$.
        This is a contradiction because by Lemma \ref{lemma:B:ell_wake_up_at_ell} the nodes at level $D$ terminate within $D+r-1$ rounds.
    \end{proof}

\subsection{Constructing an Algorithm with Property $\textit{P}(k')$ From an Algorithm with Property $\textit{P}(k)$} \label{subsection:Construction_of_Algorithm}
    Assume algorithm $\mathcal{A}$ has property $\textit{P}(k)$, for $k \geq 3$.
    Furthermore, 
    assume that the nodes at some level $\ell_0$ between $2$ and $4$ cannot distinguish between any of the messages in a set $S \subsetneq \{1,\ldots, w_{k-1}+1\}$ of size $w_{k-t_0+1}+1$ within $t_0$ rounds after waking up during the execution of $\mathcal{A}$ on $E_{12}$, for some $t_0$ between $3$ and $5$.
    % More formally, 
    % there exist a level $2 \leq \ell_0 \leq 4$, a round $3 \leq t_0 \leq 5$, and a set $S$ of size $w_{k-t_0+1}+1$
    % such that,
    % for all $0 \leq t \leq t_0$,
    % the nodes at level $\ell_0$ are in the same state
    % at the end of round $\ell_0+t$ during the execution of $\mathcal{A}$
    % with all messages  $\mu \in S$ on $E_{12}$.
    % 
    In this section, we will create algorithm $\mathcal{A}'$ from $\mathcal{A}$, with property $\textit{P}(k-t_0+2)$.
    
    We first present useful properties of algorithm $\mathcal{A}$, one for the first $\ell_0+2$ levels, one for the last $7$ levels, and then one for the levels in between.
    % 
    % We first present some useful properties of algorithm $\mathcal{A}$.
    % We start by the first $\ell_0+2$ levels. 
    % Then, we focus on the last $7$ levels. 
    % Finally, we analyse the levels in between.
    By Corollary \ref{corollary:lower_bound_at_least_i_for_level_i}, for all $\mu \in S$, 
    the nodes at level $12$ do not wake up in the first $\ell_0+t_0+2 < 12$ rounds during the execution of $\mathcal{A}$ with message $\mu$ on $E_{12}$.
    Therefore, for all $D \geq 12$ and all $\mu \in S$, 
    it follows from Lemma \ref{lemma:one_can_show_by_induction} 
    % (for graphs $E_{12}$ and $E_{\ell_0+D+6}$)
    % (where $D' = 12$, $D''=\ell_0+D+6$, $\ell'=\ell_0+t_0+3$ and $t'=\ell_0+t+2$) 
    that
    the nodes at levels $0$ through $\ell_0+1$ are in the same states
    at the end of each of the first $\ell_0+t_0+2$ rounds during the executions of $\mathcal{A}$ with message $\mu$ on graphs $E_{12}$ and $E_{\ell_0+D+6}$.
    % In particular, we get the following observation.
    \begin{observation} \label{observation:state_of_nodes_levels_0_to_ell_1_is_same}
        For all $D \geq 12$, all $\mu \in S$, all $0 \leq \ell \leq \ell_0+1$, and all $0 \leq t \leq \ell_0+t_0+2$,
        $\mathbf{s}_{\mathcal{A}}(\ell_0+D+6,\mu,\ell,t)=
        \mathbf{s}_{\mathcal{A}}(12,\mu,\ell,t)$
    \end{observation}
    
    By assumption,
    for all $0 \leq t \leq t_0$,
    the nodes at level $\ell_0$ are in the same state
    at the end of round $\ell_0+t$ during the execution of $\mathcal{A}$
    with all messages $\mu \in S$ on $E_{12}$.
    By Observation \ref{observation:state_of_nodes_levels_0_to_ell_1_is_same}, the same is true for nodes at level $\ell_0$ during the execution of $\mathcal{A}$
    on $E_{\ell_0+D+6}$, for all $D \geq 12$.
    Therefore, by Lemma \ref{lemma:distinguishability}, the same is true for nodes at greater levels.
    \begin{lemma} \label{lemma:state_of_nodes_levels_geater_than_ell_0_same}
        For all $D \geq 12$, all $0 \leq \ell \leq D$, and all $0 \leq t \leq t_0$,
        the nodes at level $\ell_0+\ell$ are in the same state
        at the end of round $\ell_0+\ell+t$ during the execution of $\mathcal{A}$
        with all messages $\mu \in S$ on $E_{\ell_0+D+6}$.
    \end{lemma}
    
    Since nodes at level $\ell_0+D$ are in the same state 
    at the end of 
    %round $\ell_0+D+t_0$
    the first $\ell_0+D+t_0$ rounds
    during the execution of $\mathcal{A}$ with all messages $\mu \in S$ on $E_{\ell_0+D+6}$, Lemma \ref{lemma:distinguishability} implies that the nodes at levels
    $\ell_0+D+1$ through $\ell_0+D+6$
    are also
    %the last $7$ levels of $E_{\ell_0+D+6}$ are 
    in the same states 
    at the end of the first $\ell_0+D+t_0$ rounds during the execution of $\mathcal{A}$ with all messages $\mu \in S$ on $E_{\ell_0+D+6}$.
    Thus, the nodes at the last $7$ levels of $E_{\ell_0+D+6}$ cannot distinguish between different messages in $S$ at the end of round $\ell_0+D+t_0$.
    
    \medskip
    
    By Corollary \ref{corollary:lower_bound_at_least_i_for_level_i}, for all $D\geq 12$ and all $1 \leq \ell' \leq D$, 
    the nodes at level $\ell = \ell_0+\ell'+t_0+1 \leq \ell_0+D+6$ do not wake up in the first $\ell_0+\ell'+t_0$ rounds during the execution of $\mathcal{A}$ on $E_{\ell_0+D+6}$.
    Thus, for all $D' > D$, 
    it follows from Lemma \ref{lemma:one_can_show_by_induction} 
    (for $t=\ell_0+\ell'+t_0-1$ and $\ell_0+\ell'+t_0 \leq \ell_0+D+t_0$)
    that,
    at the end of rounds $\ell_0+\ell'+t_0-1$ and $\ell_0+\ell'+t_0$ of the execution of $\mathcal{A}$ with message $\mu$ on $E_{\ell_0+D+6}$,
    the nodes at level $\ell_0+\ell' \leq \ell-1$ are in the same states as they are at the end of rounds $\ell_0+\ell'+t_0-1$ and $\ell_0+\ell'+t_0$ 
    of the execution of $\mathcal{A}$ with message $\mu$ on $E_{\ell_0+D'+6}$.
    In other words, for all $\ell' \geq 1$,
    there exist states 
    $\mathbf{c}_{\mathcal{A}}(\ell_0+\ell',t_0-1)$ and $\mathbf{c}_{\mathcal{A}}(\ell_0+\ell',t_0)$ 
    such that,
    the nodes at level $\ell_0+\ell'$ are in states $\mathbf{c}_{\mathcal{A}}(\ell_0+\ell',t_0-1)$ and $\mathbf{c}_{\mathcal{A}}(\ell_0+\ell',t_0)$ at the end of rounds $\ell_0+\ell'+t_0-1$ and $\ell_0+\ell'+t_0$, respectively, during the execution of $\mathcal{A}$ with message $\mu$ on $E_{\ell_0+D+6}$, for every $\mu \in S$ and every $D \geq 12$, provided $\ell_0+\ell'$ is not the last level of $E_{\ell_0+D+6}$.

\medskip

    We construct algorithm $\mathcal{A}'$ 
    so that,
    for all $D \geq 12$,
    every
    execution of $\mathcal{A}'$ with one of these messages
    on $E_D$
    simulates 
    the execution of $\mathcal{A}$ with the same message
    on $E_{\ell_0+D+6}$:
    During the execution of $\mathcal{A}'$ 
    on $E_D$, 
    the source node simulates levels $0$ through $\ell_0$ in $E_{\ell_0+D+6}$ and
    each non-source node that is not in level $D$
    simulates a node $\ell_0$ levels higher in $E_{\ell_0+D+6}$.
    The nodes at level $D$ simulate the nodes at levels $\ell_0+D$ through $\ell_0+D+6$ in $E_{\ell_0+D+6}$.
    % Nodes at level $\ell$ wake up at round $\ell$ and start simulating the nodes at level $\ell_0+\ell$ 
    % of $E_{\ell_0+D+6}$ from round $\ell_0+\ell$.

    The set of states $S_{\mathcal{A}}(\ell_0,\mu)$ will be used by the source node 
    during the execution of algorithm $\mathcal{A}'$ with message $\mu$.
    Similarly, we define a set of states that will be used by the nodes at the last level in $\mathcal{A}'$.
    Recall that, the nodes know their level and know whether that level is the last level of the graph.
    For all $D \geq 12$, let
    \begin{equation*}
    \begin{split}
        L_{\mathcal{A}}(D)=\{[\mathbf{s}_{\mathcal{A}}(\ell_0+D+6,\mu,\ell_0+D,t), \ldots, \mathbf{s}_{\mathcal{A}}(&\ell_0+D+6,\mu,\ell_0+D+6,t)] \mid \\  &\text{for all $\mu \in \{1,\ldots, w_{k-1}+1\}$ and all $t \geq 1$} \}
    \end{split}
    \end{equation*}
    be the set of all sequences of states of nodes at the last $7$ levels that can occur at the end of some round during the execution of $\mathcal{A}$ with some message on $E_{\ell_0+D+6}$.
    
    %For algorithm $\mathcal{A}'$, the set of states of the nodes at level $D$ during an execution is $L_{\mathcal{A}}(D)$.
    The transitions from  states in $L_{\mathcal{A}}(D)$ are defined similarly to the transitions from the states in $S_{\mathcal{A}}(\ell_0,\mu)$.
    % \begin{definition}  \label{L_A(D)_operation}
        The nodes in state $[\mathbf{s}_{\mathcal{A}}(\ell_0+D+6,\mu,\ell_0+D,t), \ldots, \mathbf{s}_{\mathcal{A}}(\ell_0+D+6,\mu,\ell_0+D+6,t)] \in L_{\mathcal{A}}(D)$
        perform the same operation as nodes in state $\mathbf{s}_{\mathcal{A}}(\ell_0+D+6,\mu,\ell_0+D,t)$ in $\mathcal{A}$.
    % \end{definition}
    
    % \begin{definition} \label{L_A(D)_transmit}
        If nodes in state $[\mathbf{s}_{\mathcal{A}}(\ell_0+D+6,\mu,\ell_0+D,t), \ldots, \mathbf{s}_{\mathcal{A}}(\ell_0+D+6,\mu,\ell_0+D+6,t)]$ transmit, then they transition to state $[\mathbf{s}_{\mathcal{A}}(\ell_0+D+6,\mu,\ell_0+D,t+1), \ldots, \mathbf{s}_{\mathcal{A}}(\ell_0+D+6,\mu,\ell_0+D+6,t+1)]$ immediately after transmitting.
    % \end{definition}
    This transition is well defined by Observation \ref{observation:D-6_to_D_same_seq_of_states_and_D-6_transmit}
    % , in any execution of $\mathcal{A}$ with message $\mu'$ on $E_D$ in which the nodes at levels $D-6$ through $D$ are in states $\mathbf{s}_{\mathcal{A}}(D,\mu,D-6,t), \ldots, \mathbf{s}_{\mathcal{A}}(D,\mu,D,t)$, respectively, 
    % at the end of some round $t'$
    % (i.e. $\mathbf{s}_{\mathcal{A}}(D,\mu',D-6,t')=\mathbf{s}_{\mathcal{A}}(D,\mu,D-6,t), \ldots, \mathbf{s}_{\mathcal{A}}(D,\mu',D,t')=\mathbf{s}_{\mathcal{A}}(D,\mu,D,t)$), the nodes at levels $D-6$ through $D$ transition to the states $\mathbf{s}_{\mathcal{A}}(D,\mu,D-6,t+1), \ldots, \mathbf{s}_{\mathcal{A}}(D,\mu,D,t+1)$, respectively, at the end of round $t'+1$
    % (i.e. $\mathbf{s}_{\mathcal{A}}(D,\mu',D-6,t'+1)=\mathbf{s}_{\mathcal{A}}(D,\mu,D-6,t+1), \ldots, \mathbf{s}_{\mathcal{A}}(D,\mu',D,t'+1)=\mathbf{s}_{\mathcal{A}}(D,\mu,D,t+1)$).
    % Thus, the transition in Definition \ref{L_A(D)_transmit} is well defined.
    
    Now, assume nodes in state $[\mathbf{s}_{\mathcal{A}}(\ell_0+D+6,\mu,\ell_0+D,t), \ldots, \mathbf{s}_{\mathcal{A}}(\ell_0+D+6,\mu,\ell_0+D+6,t)]$ listen.
    We separately define the transitions from this state 
    for the cases when the nodes receive a collision signal or nothing. 
    Suppose these nodes receive a collision signal while listening in state $[\mathbf{s}_{\mathcal{A}}(\ell_0+D+6,\mu,\ell_0+D,t), \ldots, \mathbf{s}_{\mathcal{A}}(\ell_0+D+6,\mu,\ell_0+D+6,t)]$.
    If nodes at level $\ell_0+D-1$ transmit at round $t+1$ during the execution of $\mathcal{A}$ with message $\mu$ on $E_{\ell_0+D+6}$, then nodes in state $[\mathbf{s}_{\mathcal{A}}(\ell_0+D+6,\mu,\ell_0+D,t), \ldots, \mathbf{s}_{\mathcal{A}}(\ell_0+D+6,\mu,\ell_0+D+6,t)]$ transition to state $[\mathbf{s}_{\mathcal{A}}(\ell_0+D+6,\mu,\ell_0+D,t+1), \ldots, \mathbf{s}_{\mathcal{A}}(\ell_0+D+6,\mu,\ell_0+D+6,t+1)]$.
    If nodes at level $\ell_0+D-1$ listen at round $t+1$ during the execution of $\mathcal{A}$ with message $\mu$ on $E_{\ell_0+D+6}$, but there exist $t'$ and $\mu'$
    such that
    the nodes at levels $\ell_0+D$ through $\ell_0+D+6$ are in states $\mathbf{s}_{\mathcal{A}}(\ell_0+D+6,\mu,\ell_0+D,t), \ldots, \mathbf{s}_{\mathcal{A}}(\ell_0+D+6,\mu,\ell_0+D+6,t)$, respectively,
    at the end of round $t'$ during the execution of $\mathcal{A}$ with message $\mu'$ on $E_{\ell_0+D+6}$ and
    the nodes at level $\ell_0+D-1$ transmit at round $t'+1$, then nodes in state $[\mathbf{s}_{\mathcal{A}}(\ell_0+D+6,\mu,\ell_0+D,t), \ldots, \mathbf{s}_{\mathcal{A}}(\ell_0+D+6,\mu,\ell_0+D+6,t)]$ transition to state $[\mathbf{s}_{\mathcal{A}}(\ell_0+D+6,\mu',\ell_0+D,t'+1), \ldots, \mathbf{s}_{\mathcal{A}}(\ell_0+D+6,\mu',\ell_0+D+6,t'+1)]$.
    By Observation \ref{observation:D-6_to_D_same_seq_of_states_and_D-7_same_op}, this transition is well defined.
    If no such $t'$ and $\mu'$ exist, then, during the construction of algorithm $\mathcal{A}'$, we make sure that the nodes at level $D$ do not receive a collision signal while listening in state $[\mathbf{s}_{\mathcal{A}}(\ell_0+D+6,\mu,\ell_0+D,t), \ldots, \mathbf{s}_{\mathcal{A}}(\ell_0+D+6,\mu,\ell_0+D+6,t)]$.

    The last case is when these nodes receive nothing while listening in state $[\mathbf{s}_{\mathcal{A}}(\ell_0+D+6,\mu,\ell_0+D,t), \ldots, \mathbf{s}_{\mathcal{A}}(\ell_0+D+6,\mu,\ell_0+D+6,t)]$.
    If nodes at level $\ell_0+D-1$ listen at round $t+1$ during the execution of $\mathcal{A}$ with message $\mu$ on $E_{\ell_0+D+6}$, then nodes in state $[\mathbf{s}_{\mathcal{A}}(\ell_0+D+6,\mu,\ell_0+D,t), \ldots, \mathbf{s}_{\mathcal{A}}(\ell_0+D+6,\mu,\ell_0+D+6,t)]$ transition to state $[\mathbf{s}_{\mathcal{A}}(\ell_0+D+6,\mu,\ell_0+D,t+1), \ldots, \mathbf{s}_{\mathcal{A}}(\ell_0+D+6,\mu,\ell_0+D+6,t+1)]$.
    Note that the nodes at level $\ell_0+D$ might receive a collision signal at round $t+1$ during the execution of $\mathcal{A}$ with message $\mu$ on $E_{\ell_0+D+6}$ from the nodes at level $\ell_0+D+1$.
    If nodes at level $\ell_0+D-1$ transmit at round $t+1$ during the execution of $\mathcal{A}$ with message $\mu$ on $E_{\ell_0+D+6}$, but there exist $t'$ and $\mu'$
    such that
    the nodes at levels $\ell_0+D$ through $\ell_0+D+6$ are in states $\mathbf{s}_{\mathcal{A}}(\ell_0+D+6,\mu,\ell_0+D,t), \ldots, \mathbf{s}_{\mathcal{A}}(\ell_0+D+6,\mu,\ell_0+D+6,t)$, respectively,
    at the end of round $t'$ during the execution of $\mathcal{A}$ with message $\mu'$ on $E_{\ell_0+D+6}$ and
    the nodes at level $\ell_0+D-1$ listen at round $t'+1$, then nodes in state $[\mathbf{s}_{\mathcal{A}}(\ell_0+D+6,\mu,\ell_0+D,t), \ldots, \mathbf{s}_{\mathcal{A}}(\ell_0+D+6,\mu,\ell_0+D+6,t)]$ transition to state $[\mathbf{s}_{\mathcal{A}}(\ell_0+D+6,\mu',\ell_0+D,t'+1), \ldots, \mathbf{s}_{\mathcal{A}}(\ell_0+D+6,\mu',\ell_0+D+6,t'+1)]$.
    By Observation \ref{observation:D-6_to_D_same_seq_of_states_and_D-7_same_op}, this transition is well defined.
    If no such $t'$ and $\mu'$ exist, then, during the construction of algorithm $\mathcal{A}'$, we make sure that the nodes at level $D$ receive a collision signal while listening in state $[\mathbf{s}_{\mathcal{A}}(\ell_0+D+6,\mu,\ell_0+D,t), \ldots, \mathbf{s}_{\mathcal{A}}(\ell_0+D+6,\mu,\ell_0+D+6,t)]$.

        \medskip
        
        % We now define the behaviour of algorithm $\mathcal{A}'$ on $E_D$, for $D \geq 12$. 
        The execution of algorithm $\mathcal{A}'$ on $E_D$ simulates the execution of $\mathcal{A}$ with the same message on $E_{\ell_0+D+6}$ as follows.
        The source node, starting from round $5$, simulates the nodes at levels $0$ through $\ell_0$ of $E_{\ell_0+D+6}$ starting at round $\ell_0+t_0+3$.
        Nodes at level $1$ wake up at round $1$ and, starting at round $5$, they simulate the nodes at level $\ell_0+1$ of $E_{\ell_0+D+6}$ starting at round $\ell_0+t_0+3$.
        Nodes at level $2\leq \ell \leq D-1$ wake up at round $\ell$ and, starting at round $\ell+2$, they simulate the nodes at level $\ell_0+\ell$ of $E_{\ell_0+D+6}$
        starting at round $\ell_0+\ell+t_0$.
        Nodes at level $D$ wake up at round $D$ and, starting at round $D+2$, they simulate the nodes at levels $\ell_0+D$ through $\ell_0+D+6$ of $E_{\ell_0+D+6}$
        starting at round $\ell_0+D+t_0$.
        
        % Starting at round $4$ of the execution of $\mathcal{A}'$ with message $\mu$ on $E_D$, 
        % simulates the nodes at levels $0$ through $\ell_0$ in the execution of $\mathcal{A}$ on $E_{\ell_0+D+6}$ starting at round $\ell_0+t_0+2$.
        % Similarly, starting at round $4$, nodes at level $1$ simulate the nodes at level $\ell_0+1$ starting at round $\ell_0+t_0+2$.
        % For all $2 \leq \ell \leq D-1$, starting at round $\ell+2$ of the execution of $\mathcal{A}'$ with message $\mu$ on $E_D$, the nodes at level $\ell$ simulate the nodes at level $\ell_0+\ell$ starting at round $\ell_0+\ell+t_0$ of the execution of $\mathcal{A}$ with message $\mu$ on $E_{\ell_0+D+6}$.
        % Starting at round $D+2$, nodes at level $D$ 
        % simulate the nodes at levels $\ell_0+D$ through $\ell_0+D+6$
        % starting at round $\ell_0+D+t_0$.
        
        %To construct $\mathcal{A}'$, we make the following modifications to $\mathcal{A}$.
        We now present a detailed description of algorithm $\mathcal{A}'$ on $E_D$, for $D \geq 12$. 
        \begin{itemize}
            \item
            The source node transmits $\mu$ at round $1$.
            At rounds $2$, $3$ and $4$, the source node listens.
            At the end of round $4$, the source node transitions to state
            $[\mathbf{s}_{\mathcal{A}}(12,\mu,0,\ell_0+t_0+2),\ldots, \mathbf{s}_{\mathcal{A}}(12,\mu,\ell_0,\ell_0+t_0+2)] \in S_{\mathcal{A}}(\ell_0,\mu)$.
            By Observation \ref{observation:state_of_nodes_levels_0_to_ell_1_is_same}, this is the same as state $[\mathbf{s}_{\mathcal{A}}(\ell_0+D+6,\mu,0,\ell_0+t_0+2),\ldots, \mathbf{s}_{\mathcal{A}}(\ell_0+D+6,\mu,\ell_0,\ell_0+t_0+2)]$, which is the sequence of states of the nodes at levels $0$ through $\ell_0$ at the end of round $\ell_0+t_0+2$ during the execution of $\mathcal{A}$ with message $\mu$ on $E_{\ell_0+D+6}$.
            At subsequent rounds, the state of the source node remains in $S_{\mathcal{A}}(\ell_0,\mu)$. 
            % The transitions from these states are given in Definitions \ref{S_A(ell,mu)_transmit}, \ref{S_A(ell,mu)_receive_col_signal} and \ref{S_A(ell,mu)_receive_nothing}.

            \item
            At round $1$, the nodes at level $1$ receive $\mu$.  
            At round $2$, the nodes at level $1$ transmit.
            At rounds $3$ and $4$, the nodes at level $1$ listen.
            At the end of round $4$, the nodes at level $1$ transition to state $\mathbf{s}_{\mathcal{A}}(12,\mu,\ell_0+1,\ell_0+t_0+2)$. 
            By Observation \ref{observation:state_of_nodes_levels_0_to_ell_1_is_same}, this is the same as state $\mathbf{s}_{\mathcal{A}}(\ell_0+D+6,\mu,\ell_0+1,\ell_0+t_0+2)$,
            which is the state of the nodes at level $\ell_0+1$ at the end of round $\ell_0+t_0+2$ during the execution of $\mathcal{A}$ with message $\mu$ on $E_{\ell_0+D+6}$.
            At subsequent rounds, the nodes at level $1$ behave like the nodes at level $\ell_0+1$ in $\mathcal{A}$.
            In particular, if the nodes at level $1$ receive a message, they act as though they receive a collision signal.

            \item 
            For  $2 \leq \ell \leq D$, the nodes at level $\ell$ transmit the round after they wake up. 
            Therefore, by Corollary \ref{corollary:local_round_i_for_ell_is_ell+i}, every node at level $\ell$ wakes up at round $\ell$.
            At round $\ell+2$,
            the nodes at level $\ell$ perform the same operation as nodes in state 
            $\mathbf{c}_{\mathcal{A}}(\ell_0+\ell,t_0-1)$,
            which is the operation that the nodes at level $\ell_0+\ell$ perform at round $\ell_0+\ell+t_0$ during the execution of $\mathcal{A}$ with message $\mu$ on $E_{\ell_0+D+6}$.
            
            If $\ell < D$, 
            then the nodes at level $\ell$ transition to state 
            $\mathbf{c}_{\mathcal{A}}(\ell_0+\ell,t_0)$ at the end of round $\ell+2$, which is the state of the nodes at level $\ell_0+\ell$ at the end of round $\ell_0+\ell+t_0$ during the execution of $\mathcal{A}$ with message $\mu$ on $E_{\ell_0+D+6}$. 
            At subsequent rounds, the nodes at level $\ell$ behave same as the nodes at level $\ell_0+\ell$ in $\mathcal{A}$.
            
            At the end of round $D+2$, the nodes at level $D$ transition to state 
            $[\mathbf{s}_{\mathcal{A}}(\ell_0+D+6,\mu,\ell_0+D,\ell_0+D+t_0), \ldots,   \mathbf{s}_{\mathcal{A}}(\ell_0+D+6,\mu,\ell_0+D+6,\ell_0+D+t_0)] \in L_{\mathcal{A}}(D)$, which is the sequence of states of the nodes at levels $\ell_0+D$ through $\ell_0+D+6$ at the end of round $\ell_0+D+t_0$ during the execution of $\mathcal{A}$ on $E_{\ell_0+D+6}$.
            This state does not depend on the message $\mu$ because
            the nodes at the last $7$ levels of $E_{\ell_0+D+6}$ cannot distinguish between different
            messages in $S$ at the end of round $\ell_0+D+t_0$ during the execution of $\mathcal{A}$ on $E_{\ell_0+D+6}$.
            At subsequent rounds, the states of the nodes at level $D$ remain in $L_{\mathcal{A}}(D)$. 
            % The transitions from these states are given in Definitions \ref{L_A(D)_transmit}, \ref{L_A(D)_receive_col_signal} and \ref{L_A(D)_receive_nothing}.
        \end{itemize}

        Let $\alpha_{\mu,D}$ be the execution of $\mathcal{A}$ with message $\mu$ on $E_{D}$ and $\alpha_{\mu,D}'$ be the execution $\mathcal{A}'$ with message $\mu$ on $E_{D}$. 
        We say that
        $\alpha_{\mu,D}'$ {\em simulates} $\alpha_{\mu,\ell_0+D+6}$ {\em for j steps}
        if $j = 0$ or $j > 0$,
        $\alpha_{\mu,D}'$ simulates $\alpha_{\mu,\ell_0+D+6}$ for $j-1$ steps, and the following relationships
        between the states of nodes in $\alpha_{\mu,D}'$ and $\alpha_{\mu,\ell_0+D+6}$
        hold:
        \begin{itemize}
            \item At the end of round $j+4$ of $\alpha_{\mu,D}'$, the source node is in state\\ $[\mathbf{s}_{\mathcal{A}}(\ell_0+D+6,\mu,0,\ell_0+t_0+j+2),\ldots, \mathbf{s}_{\mathcal{A}}(\ell_0+D+6,\mu,\ell_0,\ell_0+t_0+j+2)] \in S_{\mathcal{A}}(\ell_0,\mu)$.
            \item At the end of round $j+4$ of $\alpha_{\mu,D}'$, the nodes at level $1$ are in the same state as the nodes at level $\ell_0+1$ at the end of round $\ell_0+t_0+j+2$ of $\alpha_{\mu,\ell_0+D+6}$, i.e. they are in state $\mathbf{s}_{\mathcal{A}}(\ell_0+D+6,\mu,\ell_0+1,\ell_0+t_0+j+2)$.
            \item For all $2 \leq \ell \leq D-1$, at the end of round $\ell+j+2$ of $\alpha_{\mu,D}'$, the nodes at level $\ell$ are in the same state as the nodes at level $\ell_0+\ell$ at the end of round $\ell_0+\ell+t_0+j$ of $\alpha_{\mu,\ell_0+D+6}$.
            \item At the end of round $D+j+2$ of $\alpha_{\mu,D}'$, the nodes at level $D$ are in state 
            $[\mathbf{s}_{\mathcal{A}}(\ell_0+D+6,\mu,\ell_0+D,\ell_0+D+t_0+j),\ldots, \mathbf{s}_{\mathcal{A}}(\ell_0+D+6,\mu,\ell_0+D+6,\ell_0+D+t_0+j)] \in  L_{\mathcal{A}}(D)$.
        \end{itemize}

        We show that, for all $j \geq 0$,
        $\alpha_{\mu,D}'$ simulates $\alpha_{\mu,\ell_0+D+6}$ for $j$ steps.
        \begin{lemma} \label{lemma:reductions}
            % Suppose that, for all $D \geq 12$, all $\mu \in S$ and all $1 \leq \ell \leq D$, the nodes at level $\ell$ wake up at round $\ell$ of the execution of $\mathcal{A}$ with message $\mu$ on $E_{D}$.
            % Moreover, assume that there exist a level $2 \leq \ell_0 \leq 4$, a round $3 \leq t_0 \leq 5$ and states 
            % $\mathbf{c}_{\mathcal{A}}(\ell_0,1), \ldots, \mathbf{c}_{\mathcal{A}}(\ell_0,t_0)$ 
            % such that,
            % for all $\mu \in S$ and all $0 \leq t \leq t_0$,
            % the nodes at level $\ell_0$ are in state $\mathbf{c}_{\mathcal{A}}(\ell_0,t)$ at the end of round $\ell_0+t$ 
            % during the execution of $\mathcal{A}$ with message $\mu$ on $E_{12}$.
            % Then, for all $D \geq 12$, all $\mu \in S$ and all $1 \leq \ell \leq D$, 
            % during the execution of $\mathcal{A}'$ with message $\mu$ on $E_{D}$,
            % the nodes at level $\ell$ wake up at round $\ell$.
            % Furthermore, for all $t \geq 0$, 
            % the execution of $\mathcal{A}'$ with message $\mu$ on $E_D$
            % $(t,t_0)$-simulates
            % the execution of $\mathcal{A}$ with message $\mu$ on $E_{\ell_0+D+6}$. 
            For all $D \geq 12$, all $\mu \in S$, and all $j \geq 0$, 
            the execution of $\mathcal{A}'$ with message $\mu$ on $E_D$
            simulates
            the execution of $\mathcal{A}$ with message $\mu$ on $E_{\ell_0+D+6}$ for $j$ steps.
        \end{lemma}
        \begin{proof}
            Let $D \geq 12$ and $\mu \in S$.
            % Also, let $\alpha_{\mu,D}'$ be the execution of $\mathcal{A}'$ with message $\mu$ on $E_{D}$.
            We prove, by induction on $j \geq 0$, that $\alpha_{\mu,D}'$ simulates $\alpha_{\mu,\ell_0+D+6}$ for $j$ steps.
            Let $j \geq 1$ and assume that $\alpha_{\mu,D}'$ simulates $\alpha_{\mu,\ell_0+D+6}$ for $j-1$ steps.
            
            At the end of round $(j-1)+4$ of $\alpha_{\mu,D}'$,
            the source node is in state 
            $[\mathbf{s}_{\mathcal{A}}(\ell_0+D+6,\mu,0,\ell_0+t_0+(j-1)+2),\ldots, \mathbf{s}_{\mathcal{A}}(\ell_0+D+6,\mu,\ell_0,\ell_0+t_0+(j-1)+2)]$ and
            the nodes at level $1$ are in the same state as the nodes at level $\ell_0+1$ at the end of round $\ell_0+t_0+(j-1)+2$ of $\alpha_{\mu,\ell_0+D+6}$.
            Furthermore, at the end of round $2+(j-1)+2$ of $\alpha_{\mu,D}'$, the nodes at level $2$ are in the same states as the nodes at level $\ell_0+2$ at the end of round $(\ell_0+2)+t_0+(j-1)$ of $\alpha_{\mu,\ell_0+D+6}$.
            
            At round $j+4$ of $\alpha_{\mu,D}'$, the source node, the nodes at level $1$, and the nodes at level $2$ perform the same operation as the nodes at level $\ell_0$, $\ell_0+1$, and $\ell_0+2$, respectively, 
            at round $\ell_0+t_0+j+2$ of $\alpha_{\mu,\ell_0+D+6}$.
            Hence, the nodes at level $1$ transition to the same state as the nodes at level $\ell_0+1$ at the end of round $\ell_0+t_0+j+2$ of $\alpha_{\mu,\ell_0+D+6}$. 
            Moreover, by definition of the transitions of states in $S_{\mathcal{A}}(\ell_0,\mu)$,
            at the end of round $j+4$ of $\alpha_{\mu,D}'$, 
            the source node transitions to the state 
            $[\mathbf{s}_{\mathcal{A}}(\ell_0+D+6,\mu,0,\ell_0+t_0+j+2),\ldots, \mathbf{s}_{\mathcal{A}}(\ell_0+D+6,\mu,\ell_0,\ell_0+t_0+j+2)] \in S_{\mathcal{A}}(\ell_0,\mu)$.
        
            % If the nodes at level $1$ receive a message at round $j+4$ of $\alpha_{\mu,D}'$, the nodes at level $\ell_0+1$ receive a collision signal at round $\ell_0+t_0+j+2$ of $\alpha_{\mu,\ell_0+D+6}$.
            % However, in this case, by construction of $\mathcal{A}'$, the nodes at level $1$ act as though they received a collision signal.
            % Thus, at the end of round $j+4$ of $\alpha_{\mu,D}'$, 
            % the nodes at level $1$ transition to the same state as the nodes at level $\ell_0+1$ at the end of round $\ell_0+t_0+j+2$ of $\alpha_{\mu,\ell_0+D+6}$. 
            
            Let $2 \leq \ell \leq D-1$ and assume that, at the end of round $(\ell-1)+j+2$ of $\alpha_{\mu,D}'$, the nodes at level $\ell-1$ are in the same state as the nodes at level $\ell_0+\ell-1$ at the end of round $\ell_0+(\ell-1)+t_0+j$ of $\alpha_{\mu,\ell_0+D+6}$. 
            %Since $\alpha_{\mu,D}'$ simulates $\alpha_{\mu,\ell_0+D+6}$ for $(j-1)$ steps, ONLY FOR $j >1$
            At the end of round $\ell+(j-1)+2$ of $\alpha_{\mu,D}'$, the nodes at level $\ell$ are in the same state as the nodes at level $\ell_0+\ell$ at the end of round $\ell_0+\ell+t_0+(j-1)$ of $\alpha_{\mu,\ell_0+D+6}$.
            Thus, the nodes at levels $\ell-1$ and $\ell$
            perform the same operations
            at round $\ell+j+2$ 
            of $\alpha_{\mu,D}'$ 
            as the nodes at levels $\ell_0+\ell-1$ and $\ell_0+\ell$
            perform
            at round $\ell_0+\ell+t_0+j$ of $\alpha_{\mu,\ell_0+D+6}$.
            If $j=1$, then, by construction of $\mathcal{A}'$, the nodes at level $\ell+1$ perform the same operation
            at round $(\ell+1) +2 = \ell + j +2$ of $\alpha_{\mu,D}'$ as nodes at level $\ell_0+(\ell+1)$ perform at round $\ell_0+(\ell+1)+t_0$ of $\alpha_{\mu,\ell_0+D+6}$.
            If $j \geq 2$, 
            % $\alpha_{\mu,D}'$ simulates $\alpha_{\mu,\ell_0+D+6}$ for $(j-2)$ steps and, hence,
            the nodes at level $\ell+1$ are in the same state at the end of round $(\ell+1)+(j-2)+2$ of $\alpha_{\mu,D}'$ as the state of nodes at level $\ell_0+\ell+1$ at round $\ell_0+(\ell+1)+t_0+(j-2)$ of $\alpha_{\mu,\ell_0+D+6}$. Hence, the nodes at level $\ell+1$ perform the same operation
            at round $(\ell+1)+(j-2)+3$ of $\alpha_{\mu,D}'$ as nodes at level $\ell_0+\ell+1$ perform at round $\ell_0+(\ell+1)+t_0+(j-2)+1$ of $\alpha_{\mu,\ell_0+D+6}$.
            % 
            % Since nodes at levels $\ell-1$, $\ell$ and $\ell+1$ 
            % perform the same operations
            % at round $\ell+j+2$ 
            % of $\alpha_{\mu,D}'$ 
            % as the nodes at levels $\ell_0+\ell-1$, $\ell_0+\ell$ and $\ell_0+\ell+1$ 
            % at round $\ell_0+\ell+t_0+j$ of $\alpha_{\mu,\ell_0+D+6}$,
            Therefore, 
            the nodes at level $\ell$ transition to 
            the same state
            at the end of round $\ell+j+2$ 
            of $\alpha_{\mu,D}'$ 
            as the nodes at level $\ell_0+\ell$
            at the end of round $\ell_0+\ell+t_0+j$ of $\alpha_{\mu,\ell_0+D+6}$.
            
            At the end of round $(D-1)+j+2$ of $\alpha_{\mu,D}'$, the nodes at level $D-1$ are in the same state as the nodes at level $\ell_0+D-1$ at the end of round $\ell_0+(D-1)+t_0+j$ of $\alpha_{\mu,\ell_0+D+6}$. 
            Thus, at round $(D-1)+j+3$ of $\alpha_{\mu,D}'$, the nodes at level $D-1$ 
            %and $D$
            perform the same operation as the nodes at level $\ell_0+D-1$ at round 
            $\ell_0+(D-1)+t_0+j+1$ of $\alpha_{\mu,\ell_0+D+6}$.
            %Furthermore, since $\alpha_{\mu,D}'$ simulates $\alpha_{\mu,\ell_0+D+6}$ for $(j-1)$ steps, 
            At the end of round $D+(j-1)+2$ of $\alpha_{\mu,D}'$, the nodes at level $D$ are in state 
            $[\mathbf{s}_{\mathcal{A}}(\ell_0+D+6,\mu,\ell_0+D,\ell_0+D+t_0+(j-1)),\ldots, \mathbf{s}_{\mathcal{A}}(\ell_0+D+6,\mu,\ell_0+D+6,\ell_0+D+t_0+(j-1))]$.
            By definition of the transitions of states in $L_{\mathcal{A}}(D)$,
            at the end of round $D+j+2$ of $\alpha_{\mu,D}'$, the nodes at level $D$ transition to the state 
            $ [\mathbf{s}_{\mathcal{A}}(\ell_0+D+6,\mu,\ell_0+D,\ell_0+D+t_0+j),\ldots, \mathbf{s}_{\mathcal{A}}(\ell_0+D+6,\mu,\ell_0+D+6,\ell_0+D+t_0+j)] \in L_{\mathcal{A}}(D)$.
            
            Therefore, $\alpha_{\mu,D}'$ simulates $\alpha_{\mu,\ell_0+D+6}$ for $j$ steps.
    \end{proof}
    
    We now show that algorithm $\mathcal{A}'$ has property $\textit{P}(k')$, for some $k'< k$.
    \begin{lemma} \label{lemma:construction_of_X''}
        Suppose that there exist a level $2 \leq \ell_0 \leq 4$, a round $3 \leq t_0 \leq 5$, and a set $S \subsetneq \{1,\ldots, w_{k-1}+1\}$ of size $w_{k-t_0+1}+1$
        such that,
        for all $0 \leq t \leq t_0$,
        the nodes at level $\ell_0$ are in the same state
        at the end of round $\ell_0+t$ during the execution of $\mathcal{A}$
        with all messages $\mu \in S$ on $E_{12}$.
        Then, there exists an algorithm that has property $\textit{P}(k-t_0+2)$.
    \end{lemma}
    \begin{proof} 
        In $\mathcal{A}'$, each node transmits the round after it wakes up.
        By Corollary \ref{corollary:local_round_i_for_ell_is_ell+i}, for all $1 \leq \ell \leq D$, the nodes at level $\ell$ wake up at round $\ell$.
        Thus, $\mathcal{A}'$ satisfies the first condition of property $\textit{P}(k-t_0+2)$.
        
        Since $\mathcal{A}$ has property $\textit{P}(k)$, it follows that, for all distinct $\mu,\mu' \in S$, all $D \geq 12$ and all $2 \leq \ell' \leq (\ell_0+D+6)-1$, there exists $0 \leq t_{\ell'} \leq k-1$
        such that, 
        the nodes at level $\ell'$ are in different states at the end of round $\ell'+t_{\ell'}$ during the executions of $\mathcal{A}$ with messages $\mu$ and $\mu'$ on $E_{\ell_0+D+6}$.
        By Observation \ref{lemma:state_of_nodes_levels_geater_than_ell_0_same},
        if $\ell' \geq \ell_0$, 
        the nodes at level $\ell'$ are in the same states
        at the end of the first $\ell'+t_0$ rounds during the execution of $\mathcal{A}$
        with all messages $\mu \in S$ on $E_{\ell_0+D+6}$.
        Hence $t_{\ell'} > t_0$.
        By Lemma \ref{lemma:reductions} (for $j=t_{\ell'}-t_0$ and $\ell=\ell'-\ell_0$), 
        % if $\ell_0+2 \leq \ell' \leq \ell_0-D-1$,
        for all $2 \leq \ell'-\ell_0 \leq D-1$,
        % the nodes at level $\ell_0+\ell$,
        % at the end of round $\ell_0+\ell+t_{\ell_0+\ell}-t_0+2 > \ell_0+\ell+2$
        % during the execution of $\mathcal{A}'$ on $E_D$,
        at the end of round $\ell'-\ell_0+t_{\ell'}-t_0+2 > \ell'-\ell_0+2$ during the execution of $\mathcal{A}'$ on $E_D$,
        the nodes at level $\ell'-\ell_0$ 
        are in the same state
        as the nodes at level $\ell'$
        at the end of round 
        $\ell'+t_{\ell'}$
        during the execution of $\mathcal{A}$ with the same message on $E_{\ell_0+D+6}$.
        Thus, the nodes at level $\ell'-\ell_0$ are in different states at the end of round $\ell'-\ell_0+(t_{\ell'}-t_0+2)$ during the executions of $\mathcal{A}'$ with messages $\mu$ and $\mu'$ on $E_{D}$.
        Let $k'=k-t_0+2$.
        Since $0 \leq t_{\ell'}-t_0+2 \leq k'-1$ and $|S|=w_{k'-1}+1$, it follows that $\mathcal{A}'$ satisfies the second condition of property $\textit{P}(k')$.
        % 
        % 
        % Since $\mathcal{A}$ has property $\textit{P}(k)$, it follows that, for all distinct $\mu,\mu' \in S$, all $D \geq 12$ and all $2 \leq \ell \leq D-1$, there exists $0 \leq t_{\ell_0+\ell} \leq k-1$
        % such that, 
        % the nodes at level $\ell_0+\ell$ are in different states at the end of round $\ell_0+\ell+t_{\ell_0+\ell}$ during the executions of $\mathcal{A}$ with messages $\mu$ and $\mu'$ on $E_{\ell_0+D+6}$.
        % By Observation \ref{lemma:state_of_nodes_levels_geater_than_ell_0_same},
        % the nodes at level $\ell_0+\ell$ cannot distinguish between the values in $S$ at the end of the first $t_0$ rounds after they wake up and, hence, $t_{\ell_0+\ell} > t_0$. 
        % By Lemma \ref{lemma:reductions}, 
        % the nodes at level $\ell$,
        % at the end of round $\ell+t_{\ell_0+\ell}-t_0+2 > \ell+2$
        % during the execution of $\mathcal{A}'$ on $E_D$,
        % are in the same state as the nodes at level $\ell_0+\ell$
        % at the end of round 
        % $\ell_0+\ell+t_{\ell_0+\ell}$
        % during the execution of $\mathcal{A}$ with the same message on $E_{\ell_0+D+6}$.
        % Thus, the nodes at level $\ell$ are in different states at the end of round $\ell+t_{\ell_0+\ell}-t_0+2$ during the executions of $\mathcal{A}'$ with messages $\mu$ and $\mu'$ on $E_{D}$.
        % Since $t_{\ell_0+\ell}-t_0+2 \leq k-t_0+1$ and $|S|=w_{k-t_0+1}+1$, it follows that $\mathcal{A}'$ satisfies the second condition of property $\textit{P}(k-t_0+2)$.
    \end{proof}

\subsection{No Algorithm Has Property $\textit{P}(k)$} \label{subsection:P(K)_is_impossible}
    To obtain a contradiction, let $k\geq 3$ be the smallest value such that there exists an algorithm $\mathcal{X}$ with property $\textit{P}(k)$. 
    
    % Since the nodes at level $3 \leq \ell \leq D$ wake up at round $\ell$
    % during the execution of $\mathcal{X}$ on $E_D$,
    % the nodes at level $\ell-1$ transmit at round $\ell$.
    % Thus, 
    % during the execution of $\mathcal{X}$ on $E_D$,
    % %for all $D \geq 12$ and 
    % for all $2 \leq \ell \leq D-1$,
    % the nodes at level $\ell$ transmit after waking up (in state $\mathbf{c}_{\mathcal{X}}(\ell,0)$) and transition to the
    % same state, which we call $\mathbf{c}_{\mathcal{X}}(\ell,1)$.
     
    % \begin{observation} \label{observation:q_1_at_ell+1}
    %     For all $D \geq 12$, all $2 \leq \ell \leq D-1$ and all $\mu \in \{1,\ldots,m\}$, 
    %     during the execution of $\mathcal{X}$ with message $\mu$ on $E_{D}$,
    %     the nodes at level $\ell$ 
    %     are in states $\mathbf{c}_{\mathcal{X}}(\ell,0)$ and $\mathbf{c}_{\mathcal{X}}(\ell,1)$ 
    %     at the end of rounds $\ell$ and $\ell+1$,
    %     respectively.
    % \end{observation}
   The first condition of property $P(k)$ implies that, for all $\ell \geq 2$, nodes in state $\mathbf{c}_{\mathcal{X}}(\ell,0)$ transmit and transition to the same state. 
   We call this state $\mathbf{c}_{\mathcal{X}}(\ell,1)$.
    %Let $\ell \geq 2$. 
    % If nodes in state $\mathbf{c}_{\mathcal{X}}(\ell,1)$ transmit, let $\mathbf{c}_{\mathcal{X}}(\ell,2)$ be the state to which they transition after transmitting. 
    % Otherwise, let $\mathbf{c}_{\mathcal{X}}(\ell,2)$ be the state to which they transition after receiving a collision signal. 
    % Since nodes at level $\ell+1$ transmit at round $\ell+2$,
    % in state $\mathbf{c}_{\mathcal{X}}(\ell+1,0)$
    % the nodes at level $\ell$ receive a collision signal at round $\ell+2$ if they are listening.
    % in state $\mathbf{c}_{\mathcal{X}}(\ell,1)$ 
    % 
    % Thus, for all $D \geq 12$, all $\mu \in \{1,\ldots,m\}$, and all $2 \leq \ell \leq D-2$, the nodes at level $\ell$ are in state $\mathbf{c}_{\mathcal{X}}(\ell,2)$ at the end of round $\ell+2$ during the execution of $\mathcal{X}$ with message $\mu$ on $E_{D}$.
    % 
    If nodes in state $\mathbf{c}_{\mathcal{X}}(\ell,1)$ transmit, let $\mathbf{c}_{\mathcal{X}}(\ell,2)$ be the state to which they transition after transmitting. 
    Otherwise, let $\mathbf{c}_{\mathcal{X}}(\ell,2)$ be the state to which they transition after receiving a collision signal. 
    If $\ell+1$ is not the last level in the graph,
    the nodes at level $\ell+1$ transmit at round $\ell+2$,
    so
    % Since nodes at level $\ell+1$ transmit at round $\ell+2$
    % provided $\ell+1$ is not the last level in the graph,
    the nodes at level $\ell$ receive a collision signal at round $\ell+2$ if they are listening.
    % in state $\mathbf{c}_{\mathcal{X}}(\ell,1)$.
    Thus, if $\ell \leq D-2$, the nodes at level $\ell$ are in state $\mathbf{c}_{\mathcal{X}}(\ell,2)$ at the end of round $\ell+2$ during the execution of $\mathcal{X}$ on $E_{D}$.

    \begin{lemma} \label{lemma:k_values:r_at_least_4_helper}
        For all $D \geq 12$, all $2 \leq \ell \leq D-2$ and all $0 \leq t \leq 2$,
        the nodes at level $\ell$ are in the same states
        at the end of round $\ell+t$ during the execution of $\mathcal{X}$ with all messages $\mu \in \{1,\ldots, w_{k-1}+1\}$ on $E_{D}$.
    \end{lemma}
    \noindent
    Hence, $\mathcal{X}$ does not satisfy the second condition of property $\textit{P}(3)$.
    Since $\mathcal{X}$ has property $\textit{P}(k)$, it follows 
    %from Lemma \ref{lemma:k_values:r_at_least_4_helper}
    that $k \geq 4$.
    
    \medskip
    
    We focus on nodes at levels $2$, $3$, and $4$ and get an upper bound on the number of messages that they cannot distinguish from one another three rounds after they wake up.
    % Next, we bound the number of messages in $\{1,\ldots,m\}$ for which the nodes at level $2 \leq \ell \leq 4$ can be in the same state at the end of round $\ell+3$.
    \begin{lemma}  \label{lemma:reduction_1}
        For $2 \leq \ell \leq 4$,
        there are at most $w_{k-2}$ different messages $\mu \in \{1,\ldots, w_{k-1}+1\}$ such that, during the execution of $\mathcal{X}$ with message $\mu$ on $E_{12}$, the nodes at level $\ell$ are in the same state at the end of round $\ell+3$. 
    \end{lemma}
    \begin{proof}
        To get a contradiction, 
        suppose there exist a level $2 \leq \ell' \leq 4$,
        a set $S \subseteq \{1,\ldots, w_{k-1}+1\}$ with $|S| \geq w_{k-2}+1$,
        and a state $q$ such that, for all $\mu \in S$,
        the nodes at level $\ell'$ are in state $q$ at the end of round $\ell'+3$ during the execution of $\mathcal{X}$ with message $\mu$ on $E_{12}$.
        
        During the execution of $\mathcal{X}$ with all messages $\mu \in S$ on $E_{12}$,
        the nodes at level $\ell'$ are in the same state at the end of round $\ell'+t$ for $0 \leq t \leq 2$ and 
        they are in state $q$ at the end of round $\ell'+3$.
        Thus, the precondition of Lemma \ref{lemma:construction_of_X''} holds for $\ell_0=\ell'$ and $t_0=3$.
        % By Lemma \ref{lemma:k_values:r_at_least_4_helper},
        % the second precondition of Lemma \ref{lemma:construction_of_X''} holds for $t'=3$.
        % Since $|S| \geq w_{k-t'+1}+1$ for $t'=3$,
        Hence, there is an algorithm  with property $\textit{P}(k-1)$.
        By definition of $k$, no such algorithm exists.
        Thus, $|S| \leq w_{k-2}$.
    \end{proof}

    Recall that, for  $2 \leq \ell \leq 4$, the nodes at level $\ell$  are in state $\mathbf{c}_{\mathcal{X}}(\ell,2)$ at the end of round $\ell+2$ during the execution of $\mathcal{X}$ for every message $\mu \in \{1,\ldots, w_{k-1}+1\}$ on $E_{12}$.
    If nodes in state $\mathbf{c}_{\mathcal{X}}(\ell,2)$ transmit, 
    nodes at level $\ell$ transition to the same state at the end of round $\ell+3$.
    % during the execution of $\mathcal{X}$ for every message $\mu \in \{1,\ldots,m\}$ on $E_{12}$. 
    By Lemma \ref{lemma:reduction_1}, there are at most $w_{k-2}<w_{k-1}+1$ different messages such that the nodes at level $\ell$ are in the 
    same state at the end of round $\ell+3$.
    Thus, nodes in state $\mathbf{c}_{\mathcal{X}}(\ell,2)$ listen.
    Nodes at levels $2$ and greater can receive only a collision signal or nothing while listening during any execution on $E_{12}$.
    % This is shown in Lemma \ref{lemma:2_and_greater_receive_no_messages}.
    Thus, there are two states that are reachable by nodes at level $\ell$ from state $\mathbf{c}_{\mathcal{X}}(\ell,2)$, we call these states
    $\mathbf{c}_{\mathcal{X}}(\ell,3)$ and $\mathbf{c}_{\mathcal{X}}'(\ell,3)$.
    Next, we show that nodes in state $\mathbf{c}_{\mathcal{X}}(\ell,3)$ perform different operation than nodes in state $\mathbf{c}_{\mathcal{X}}'(\ell,3)$, for $\ell=2,3$.
    \begin{lemma} \label{lemma:k_Values:different_operations_at_q_3_and_q'_3}
        For $\ell=2,3$, nodes in exactly one of the states $\mathbf{c}_{\mathcal{X}}(\ell,3)$ and $\mathbf{c}_{\mathcal{X}}'(\ell,3)$ transmit.
    \end{lemma} 
    \begin{proof}   
        To obtain a contradiction, assume that for $\ell=2$ or $\ell=3$, nodes in  both $\mathbf{c}_{\mathcal{X}}(\ell,3)$ and $\mathbf{c}_{\mathcal{X}}'(\ell,3)$ transmit or  nodes in both these states listen.
        Let $\mu,\mu' \in \{1,\ldots, w_{k-1}+1\}$.
        Let $\alpha_{\mu}$ and $\alpha_{\mu'}$ be the executions of $\mathcal{X}$ with messages $\mu$ and $\mu'$ on $E_{12}$, respectively.
        At the end of round $\ell+3$ of $\alpha_{\mu}$ and $\alpha_{\mu'}$,
        the nodes at level $\ell$ are in one of the states $\mathbf{c}_{\mathcal{X}}(\ell,3)$ and $\mathbf{c}_{\mathcal{X}}'(\ell,3)$, 
        the nodes at level $\ell+1$ are in state $\mathbf{c}_{\mathcal{X}}(\ell+1,2)$, and
        the nodes at level $\ell+2$ are in state $\mathbf{c}_{\mathcal{X}}(\ell+2,1)$.
        Since nodes in both of the states $\mathbf{c}_{\mathcal{X}}(\ell,3)$ and $\mathbf{c}_{\mathcal{X}}'(\ell,3)$ transmit or nodes in both of these states listen,
        the nodes at level $\ell+1$ cannot distinguish between executions $\alpha_{\mu}$ and $\alpha_{\mu'}$ at the end of round $\ell+4$.
        Hence, the nodes at level $\ell+1$ are in the same state at the end of round $\ell+4$ of $\alpha_{\mu}$ and $\alpha_{\mu'}$.
        Therefore, for all $\mu \in \{1,\ldots, w_{k-1}+1\}$, during the execution of $\mathcal{X}$ with message $\mu$ on $E_{12}$, the nodes at level $\ell+1$ are in the same state at the end of round $(\ell+1)+3$.
        However, this contradicts Lemma \ref{lemma:reduction_1} because $w_{k-1}+1>w_{k-2}$.
    \end{proof} 

    For $\ell=2,3$, suppose nodes in state $\mathbf{c}_{\mathcal{X}}(\ell,3)$ transmit and nodes in state $\mathbf{c}_{\mathcal{X}}'(\ell,3)$ listen.
    Let $U \subseteq \{1,\ldots, w_{k-1}+1\}$ be the set of messages such that nodes at level $2$ are in state $\mathbf{c}_{\mathcal{X}}(2,3)$ at the end of round $5$ during the executions of $\mathcal{X}$ with message $\mu \in U$ on $E_{12}$ and let $U'=\{1,\ldots, w_{k-1}+1\} \setminus U$.
    By Lemma \ref{lemma:reduction_1}, $|U| \leq w_{k-2}$.
    Since $|U|+|U'|=w_{k-1}+1$,
    it follows that $|U'| \geq w_{k-1}+1 - w_{k-2} = w_{k-4}+1$.
    Similarly, $|U'| \leq w_{k-2}$ and $|U| \geq w_{k-4}+1$.
    
    % Recall that, the nodes at level $6$ do not wake up in the first $5$ rounds during the execution of $\mathcal{X}$ on $E_{12}$.
    % Since nodes at level $2$ are in state $\mathbf{c}_{\mathcal{X}}(2,3)$ at the end of round $5$ during the execution of $\mathcal{X}$ with all messages $\mu \in U$ on $E_{12}$, 
    % it follows from Lemma \ref{lemma:one_can_show_by_induction} that 
    % the nodes at level $2$ are in state $\mathbf{c}_{\mathcal{X}}(2,3)$ at the end of round $5$ during the execution of $\mathcal{X}$ with all messages $\mu \in U$ on $E_{D}$, for all $D \geq 12$.
    By assumption, nodes in state $\mathbf{c}_{\mathcal{X}}(2,3)$ transmit and, hence, there can be only one state to which nodes in state $\mathbf{c}_{\mathcal{X}}(2,3)$ transition.
    Thus, the nodes at level $2$ cannot distinguish between the values in $U$ 
    at the end of the first $4$ rounds after they wake up on $E_{12}$.
    % %at the end of each of the first $6$ rounds.
    % By Lemma \ref{lemma:distinguishability}, the nodes at level $\ell \geq 3$ also cannot distinguish between the values in $U$ at the end of the first $4$ rounds after they wake up.
    % %at the end of each of the first $\ell + 4$ rounds.
    % This proves the following result.
    \begin{lemma} \label{lemma:k_values:r_at_least_6_helper}
        % For all $D \geq 12$, all $2 \leq \ell \leq D$ and all $0 \leq t \leq 4$,
        % the nodes at level $\ell$ are in the same states
        % at the end of round $\ell+t$ during the execution of $\mathcal{X}$ with all messages $\mu \in U$ on $E_{D}$.
        For all $0 \leq t \leq 4$,
        the nodes at level $2$ are in the same states
        at the end of round $2+t$ during the execution of $\mathcal{X}$ with all messages $\mu \in U$ on $E_{12}$.
    \end{lemma}
    \noindent
    Hence, $\mathcal{X}$ does not satisfy the second condition of property $\textit{P}(5)$.
    Since $\mathcal{X}$ has property $\textit{P}(k)$, it follows 
    %from Lemma \ref{lemma:k_values:r_at_least_4_helper}
    that $k \geq 6$.
    
    \medskip
    
   Now, we show that, during the execution of $\mathcal{X}$ on $E_{12}$ with message $\mu \in U$, 
    the nodes at level $3$ listen at round $7$.
    % Now, we show that there is no $\mu \in \{1,\ldots,m\}$ such that nodes at level $2$ transmit at round $6$ and nodes at level $3$ transmit at round $7$ during the execution of $\mathcal{X}$ with message $\mu$ on $E_{12}$.
    \begin{lemma}   \label{lemma:reduction_2}
        Nodes in state $\mathbf{c}_{\mathcal{X}}(3,2)$ 
        % listen and 
        transition to state 
        $\mathbf{c}_{\mathcal{X}}'(3,3)$ after receiving a collision signal.
    \end{lemma}
    \begin{proof}
        To obtain a contradiction, assume that nodes in state $\mathbf{c}_{\mathcal{X}}(3,2)$ transition to state $\mathbf{c}_{\mathcal{X}}(3,3)$ after receiving a collision signal. 
        We will construct an algorithm that has property $\textit{P}(k-3)$.
        By definition of $k$, no such algorithm exists.
        Thus, we obtain a contradiction.
        
        For all $\mu \in U$,
        during the execution of $\mathcal{X}$ with message $\mu$ on $E_{12}$,
        at the end of round $5$,
        the nodes at level $2$ are in state $\mathbf{c}_{\mathcal{X}}(2,3)$ 
        and 
        the nodes at level $3$ are in state $\mathbf{c}_{\mathcal{X}}(3,2)$.
        Nodes in state $\mathbf{c}_{\mathcal{X}}(2,3)$ transmit 
        and nodes in state $\mathbf{c}_{\mathcal{X}}(3,2)$ listen.
        Thus, by assumption, the nodes at level $3$ are in state $\mathbf{c}_{\mathcal{X}}(3,3)$ at the end of round $6$ and, hence, they transmit at round $7$.
        
        Let $\mathbf{c}_{\mathcal{X}}(2,4)$ be the state to which nodes in state $\mathbf{c}_{\mathcal{X}}(2,3)$ transition after transmitting.
        If nodes in state $\mathbf{c}_{\mathcal{X}}(2,4)$ transmit, then let $\mathbf{c}_{\mathcal{X}}(2,5)$ be the state to which they transition after transmitting.
        If nodes in state $\mathbf{c}_{\mathcal{X}}(2,4)$ listen, then let $\mathbf{c}_{\mathcal{X}}(2,5)$ be the state to which they transition after receiving a collision signal. 
        For all $\mu \in U$,
        during the execution of $\mathcal{X}$ with message $\mu$ on $E_{12}$,
        % the nodes at level $3$ are in state $\mathbf{c}_{\mathcal{X}}(3,3)$ at the end of round $6$ and they transmit at round $7$.
        since the nodes at level $3$ transmit at round $7$, the nodes at level $2$ receive a collision signal at round $7$ if they are listening and, hence, they are in
        state $\mathbf{c}_{\mathcal{X}}(2,5)$ at the end of round 7, whether they listen or transmit at round 7.
        Since nodes at level $2$ are in state $\mathbf{c}_{\mathcal{X}}(2,t)$ at the end of round $2+t$ for $0 \leq t \leq 2$, during the execution of $\mathcal{X}$ with any message $\mu \in \{1,\ldots, w_{k-1}+1\}$ on $E_{12}$, it follows that, 
        %for all $\mu \in U$ and 
        for all $0 \leq t \leq 5$,
        the nodes at level $2$ are in state $\mathbf{c}_{\mathcal{X}}(2,t)$ at the end of round $2+t$ during the execution of $\mathcal{X}$ with
        any message $\mu \in U$ on $E_{12}$.
        Thus, the precondition of Lemma \ref{lemma:construction_of_X''} holds for $\ell_0=2$ and $t_0=5$.
        % By Lemma \ref{lemma:k_values:r_at_least_6_helper},
        % the second precondition of Lemma \ref{lemma:construction_of_X''} holds for $t'=5$ and $\ell'=2$.
        Hence, there is an algorithm with property $\textit{P}(k-3)$.
    \end{proof}

    % Next, we focus on the nodes at level $4$.
    We know that at the end of round $5$ nodes at level $4$ are in state $\mathbf{c}_{\mathcal{X}}(4,1)$ and at the end of round $6$ they are in state $\mathbf{c}_{\mathcal{X}}(4,2)$. 
    At round $6$, either they listen and receive a collision signal or they transmit.
    Now we show that they never transmit at round $6$.
    %Now we show that the former is always the case.
    \begin{lemma} \label{lemma:q_4_listen}
        % For all $\mu \in \{1,\ldots,m\}$, the nodes at level $4$ listen at round $6$
        % during the execution of $\mathcal{X}$ with message $\mu$ on $E_{12}$.
        Nodes in state $\mathbf{c}_{\mathcal{X}}(4,1)$ listen.
    \end{lemma}
    \begin{proof}
        To obtain a contradiction, suppose nodes in state $\mathbf{c}_{\mathcal{X}}(4,1)$ transmit.
        For all $\mu \in \{1,\ldots, w_{k-1}+1\}$, at the end of round $5$ during the execution of $\mathcal{X}$ with message $\mu$ on $E_{12}$, 
        the nodes at level $3$ are in state $\mathbf{c}_{\mathcal{X}}(3,2)$ and the nodes at level $4$ are in state $\mathbf{c}_{\mathcal{X}}(4,1)$.
        Since nodes in state $\mathbf{c}_{\mathcal{X}}(4,1)$ transmit, the nodes at level $4$ transmit at round $6$ and the nodes at level $3$ receive a collision signal.
        Therefore, for all $\mu \in \{1,\ldots, w_{k-1}+1\}$, during the execution of $\mathcal{X}$ with message $\mu$ on $E_{12}$, the nodes at level $3$ are in state $\mathbf{c}_{\mathcal{X}}'(3,3)$ at the end of round $6$.
        However, this contradicts Lemma \ref{lemma:reduction_1} because $w_{k-1}+1>w_{k-2}$.
    \end{proof}
    
    By definition of $U'$, for all $\mu \in U'$, at the end of round $5$ during the execution of $\mathcal{X}$ with message $\mu$ on $E_{12}$, 
    the nodes at level $2$ are in state $\mathbf{c}_{\mathcal{X}}'(2,3)$,
    the nodes at level $3$ are in state $\mathbf{c}_{\mathcal{X}}(3,2)$, 
    and the nodes at level $4$ are in state $\mathbf{c}_{\mathcal{X}}(4,1)$.
    Since nodes in state $\mathbf{c}_{\mathcal{X}}'(2,3)$ listen, the nodes at level $2$ listen at round $6$.
    By Lemma \ref{lemma:q_4_listen}, the nodes at level $4$ listen at round $6$.
    Thus, the nodes at level $3$ listen in state $\mathbf{c}_{\mathcal{X}}(3,2)$ at round $6$ and receive nothing.
    % The nodes at level $3$ receive nothing while listening in state $\mathbf{c}_{\mathcal{X}}(3,2)$ at round $6$.
    By Lemma \ref{lemma:reduction_2}, nodes in state $\mathbf{c}_{\mathcal{X}}(3,2)$ transition to state $\mathbf{c}_{\mathcal{X}}'(3,3)$ after receiving a collision signal and transition to state $\mathbf{c}_{\mathcal{X}}(3,3)$ after receiving nothing.
    Therefore, we have the following observation.
    
    \begin{observation} \label{observation:k_values:V'_c(3,3)}
        For all $\mu \in U'$, the nodes at level $3$ are in state $\mathbf{c}_{\mathcal{X}}(3,3)$ at the end of round $6$ during the execution of $\mathcal{X}$ with message $\mu$ on $E_{12}$.
    \end{observation}
    
    % We show (using Lemma \ref{lemma:one_can_show_by_induction}) that,
    % for all $D \geq 12$ and all $\mu \in U'$, during the execution of $\mathcal{X}$ with message $\mu$ on $E_{D}$, the nodes at level $3$ are in state $\mathbf{c}_{\mathcal{X}}(3,3)$ at the end of round $6$.
    % Since nodes at level $3$ are in state $\mathbf{c}_{\mathcal{X}}(3,t)$ at the end of round $3+t$ for $0 \leq t \leq 2$, it follows that
    % the nodes at level $3$ cannot distinguish between the values in $U'$
    % at the end of the first $3$ rounds after they wake up.
    % Using Lemma \ref{lemma:distinguishability}, we also show that the nodes at level $\ell \geq 4$ cannot distinguish between the values in $U'$ at the end of the first $3$ rounds after they wake up.
    % This proves the next result.
    % \begin{lemma} \label{lemma:k_values:r_at_least_6_helper2}
    %     For all $D \geq 12$, all $3 \leq \ell \leq D$ and all $0 \leq t \leq 3$,
    %     the nodes at level $\ell$ are in the same states
    %     at the end of round $\ell+t$ during the execution of $\mathcal{X}$ with all messages $\mu \in U'$ on $E_{D}$.
    % \end{lemma}

    Now we improve the upper bound on the cardinality of $U$ and $U'$. 
    \begin{lemma}  \label{lemma:reduction_3}
        $|U|,|U'| \leq  w_{k-3}$.
    \end{lemma}
    \begin{proof}
        To obtain a contradiction, suppose $|U| \geq w_{k-3} + 1$ or $|U'| \geq w_{k-3} + 1$.
        % We will construct an algorithm  that has property $\textit{P}(r-2)$.
        % By definition of $r$, no such algorithm exists.
        By definition of $U$, for all $\mu \in U$, the nodes at level $2$ are in state $\mathbf{c}_{\mathcal{X}}(2,3)$ at the end of round $5$ during the execution of $\mathcal{X}$ with message $\mu$ on $E_{12}$.
        By Observation \ref{observation:k_values:V'_c(3,3)}, for all $\mu \in U'$, the nodes at level $3$ are in state $\mathbf{c}_{\mathcal{X}}(3,3)$ at the end of round $6$ during the execution of $\mathcal{X}$ with message $\mu$ on $E_{12}$.
        
        Let $S$ be one of the sets $U$ and $U'$.
        If $S=U$, then let $\ell'=2$. Otherwise, let $\ell'=3$.
        Since nodes in state $\mathbf{c}_{\mathcal{X}}(\ell',3)$ transmit, there can be only one state to which nodes in state $\mathbf{c}_{\mathcal{X}}(\ell',3)$ transition after transmitting. 
        We call this state $\mathbf{c}_{\mathcal{X}}(\ell',4)$.
        Since nodes at level $\ell'$ are in state $\mathbf{c}_{\mathcal{X}}(\ell',t)$ at the end of round $\ell'+t$ for $0 \leq t \leq 2$, during the execution of $\mathcal{X}$ with any message $\mu \in \{1,\ldots, w_{k-1}+1\}$ on $E_{12}$, it follows that,
        for all $0 \leq t \leq 4$,
        the nodes at level $\ell'$ are in state $\mathbf{c}_{\mathcal{X}}(\ell',t)$ at the end of round $\ell'+t$ during the execution of $\mathcal{X}$ with 
        any message $\mu \in S$ on $E_{12}$.
        Thus, the precondition of Lemma \ref{lemma:construction_of_X''} holds for $\ell_0=\ell'$ and $t_0=4$.
        % By Lemmas \ref{lemma:k_values:r_at_least_6_helper} and \ref{lemma:k_values:r_at_least_6_helper2}
        % the second precondition of Lemma \ref{lemma:construction_of_X''} holds for $t'=4$.
        Hence, there is an algorithm with property $\textit{P}(k-2)$.
        By definition of $k$, no such algorithm exists.
        Thus, $|S| \leq  w_{k-3}$.
     \end{proof}

    Recall that $|U|+|U'|=w_{k-1}+1$.
    Since $|U|,|U'| \leq  w_{k-3}$, it follows that $|U|+|U'| \leq w_{k-3}+w_{k-3} = w_{k-3}+w_{k-4}+w_{k-6}$.
    Note that, $w_0,w_1,\ldots,$ is a non-decreasing sequence and, hence, $w_{k-6} \leq  w_{k-5}$.
    Thus, $|U|+|U'| \leq w_{k-3}+w_{k-4}+w_{k-5} = w_{k-2}+w_{k-4} = w_{k-1}$.
    This is a contradiction.
    % because $U \cup U'= \{1,\ldots, m=w_{k-1}+1\}$.
    % Therefore, for all $r \geq 1$, 
    % there is no algorithm that has property $\textit{P}(r)$. 
    % This shows that algorithm $\mathcal{B}'$ does not exist.
    Thus, we conclude the proof of Theorem \ref{theorem:main_lower_bound}.

\section{Future Work}
    In this paper, we presented an improved algorithm to broadcast a message from a finite set of values that works on the beeping model.
    We proved an exactly matching lower bound in the radio broadcast model with collision detection.
    This shows that the ability to send
    arbitrarily long messages instead of just a collision signal does not improve the round complexity.
    
    Our algorithm relies on an encoding mechanism that requires the set of possible messages to be finite and known in advance.
    One way to extend our algorithm to handle an infinite set of possible messages
    is to have the source node first broadcast $r$ using another algorithm 
    (for example, using {\em Beep Waves} to broadcast a self-delimiting encoding of $r$
    %(such as the one presented in Czumaj and Davies
    \cite{CZUMAJ2019}). %perhaps this could be done recursively). 
    A natural question is whether this approach
    or a recursive version of this approach
    is optimal.
    % we can design an algorithm that has exactly optimal round complexity when the set of possible messages is infinite.
    % One possible approach for the source node to first broadcast $r$ using a self delimiting encoding such as the one presented in Czumaj and Davies \cite{CZUMAJ2019}. Then, the source node broadcast the actual message using the algorithm we presented.

    A variant of the broadcast problem is the acknowledged broadcast problem in which the source node needs to eventually be informed that all nodes have learned the message.
    Chlebus, G{\k{a}}sieniec, Gibbons, Pelc, and Rytter \cite{Bogdan2002} provided a deterministic acknowledged broadcast algorithm in the radio broadcast model with collision detection, assuming that each node has a distinct identifier.
    A possible approach to solve the acknowledged broadcast problem in the anonymous beeping model
    is to extend our approach by modifying the algorithm as follows.
    The extension relies on the fact that, in our algorithm, all nodes at each level learn the last bit of the message at the same round, which is one round later than the nodes in the previous level.
    % in anonymous radio broadcast model with collision detection is to extend our algorithm as follows. 
    % We only modify the algorithm for nodes that have a neighbour at a higher level as follows.
    Each node $v$ that is not connected to a node at a higher level terminates immediately after learning the last bit of the message.
    Now consider any node $v$ at level $\ell \geq 1$ that is connected to a node at level $\ell+1$.
    Node $v$ starts executing the acknowledgment process one round after relaying the last bit of the message to its neighbours at level $\ell+1$.
    It repeatedly executes phases of three rounds until it receives an acknowledgement from nodes at level $\ell+1$, which causes it to terminate.
    In the first round, $v$ beeps to inform the nodes at level $\ell-1$ that it has not yet received an acknowledgement. 
    Once $v$ has received an acknowledgement, 
    it terminates, so it does not beep in the first round of this phase. 
    This serves as an acknowledgement for the nodes at level $\ell-1$.
    In the second round, $v$ listens to detect whether the nodes at level $\ell+1$ have received an acknowledgement.
    Specifically, if $v$ receives nothing, it means that the nodes at level $\ell+1$ have received the acknowledgement.
    Otherwise, $v$ receives a beep, in which case it idles for one round and, then, starts executing the next phase.
    An interesting extension of our paper would be to see if this modified algorithm has optimal round complexity for the acknowledged broadcast problem in anonymous radio networks with collision detection. 
    
    In addition to round complexity, energy complexity has also received considerable attention in radio broadcast models.
    Energy complexity is generally defined as
    the maximum, over all nodes, of the number of rounds in which the node transmits or listens.
    Chang, Dani, Hayes, He, Li, Pettie \cite{chang2018energy} provided broadcast algorithms in the radio broadcast model with collision detection.
    In their model, they assumed that nodes have distinct identifiers and are given the number of nodes in the network, the maximum number of neighbours of any node in the network, and the maximum distance between any two nodes in the network.
    By modifying our encoding mechanism, we can improve the energy complexity of our algorithm. 
    In particular, instead of using set $W_i$ we can use the set $W_{i}'$, which is defined as follows.
    Let $W_{1}' = \{0\}$, $W_{2}' = \{1\}$, $W_{3}' = \{0,1\}$, $W_{4}' = \{00,01,1\}$ and, 
    for all $i \geq 5$, $W_{i}'= 0 \cdot W_{i-1}' \cup 1 \cdot W_{i-2}'$. 
    For all $i\geq 1$, let $w_{i}'=|W_i|$.
    Then, $w_{1}'=w_{2}'=1$ and, for all $i\geq 3$, $w_{i}'=w_{i-1}'+ w_{i-2}'$.
    Note that, $w_{1}', w_{2}', w_{3}', \ldots$ is the Fibonacci sequence.
    The optimal energy complexity of the broadcast problem remains open.
    
    % \begin{itemize}%     
    %     \item Strengthen the lower bound model to Strong-CD, i.e. a node that transmits receives its own message if none of its neighbours transmit at that round or it receives a collision signal if at least one of its neighbours transmits at that round. 
    %         This can be done by adding an edge between the nodes at each level of $E_D$.
    % \end{itemize}

\section*{Acknowledgments}
I would like thank my supervisor, Faith Ellen, for her patience and support throughout this project. 
Her insightful feedback and guidance brought my work to a much higher level.
I would also like to thank the anonymous reviewers for their time and helpful comments.

\printbibliography

@article{Ellen2019,
  title={Constant-length labeling schemes for deterministic radio broadcast},
  author={Ellen, Faith and Gorain, Barun and Miller, Avery and Pelc, Andrzej},
  journal={ACM Transactions on Parallel Computing},
  volume={8},
  number={3},
  pages={1--17},
  year={2021},
  publisher={ACM New York, NY}
}

@article{pelc2003,
  title={Broadcasting in undirected ad hoc radio networks},
  author={Kowalski, Dariusz and Pelc, Andrzej},
  journal={Distributed Computing},
  volume={18},
  number={1},
  pages={43--57},
  year={2005},
  publisher={Springer}
}

@BOOK{narayanaBook,
author = {Narayana Pandita},
title = {Ganita Kaumudi},
year = {1356},
}

@article{Czumaj2003,
  title={Broadcasting algorithms in radio networks with unknown topology},
  author={Czumaj, Artur and Rytter, Wojciech},
  journal={Journal of Algorithms},
  volume={60},
  number={2},
  pages={115--143},
  year={2006},
  publisher={Elsevier}
}

@article{BARYEHUDA1992104,
  title={On the time-complexity of broadcast in multi-hop radio networks: An exponential gap between determinism and randomization},
  author={Bar-Yehuda, Reuven and Goldreich, Oded and Itai, Alon},
  journal={Journal of Computer and System Sciences},
  volume={45},
  number={1},
  pages={104--126},
  year={1992},
  publisher={Elsevier}
}

@article{Alon1991,
  title={A lower bound for radio broadcast},
  author={Alon, Noga and Bar-Noy, Amotz and Linial, Nathan and Peleg, David},
  journal={Journal of Computer and System Sciences},
  volume={43},
  number={2},
  pages={290--298},
  year={1991},
  publisher={Elsevier}
}

@article{wave_expansion_approach,
  title={The wave expansion approach to broadcasting in multihop radio networks},
  author={Chlamtac, Imrich},
  journal={IEEE Transactions on Communications},
  volume={39},
  number={3},
  pages={426--433},
  year={1991},
  publisher={IEEE}
}

@article{Mansour1993,
  title = {An $\Omega(D \cdot \log{\frac{N}{D}})$ Lower Bound for Broadcast in Radio Networks},
  author={Kushilevitz, Eyal and Mansour, Yishay},
  journal={SIAM journal on Computing},
  volume={27},
  number={3},
  pages={702--712},
  year={1998},
  publisher={SIAM}
}

@inproceedings{Bogdan2000,
  title={Deterministic radio broadcasting},
  author={Chlebus, Bogdan and G{\k{a}}sieniec, Leszek and {\"O}stlin, Anna and Robson, John Michael},
  booktitle={International Colloquium on Automata, Languages, and Programming},
  pages={717--729},
  year={2000},
  organization={Springer}
}

@article{Bogdan2002,
  title={Deterministic broadcasting in ad hoc radio networks},
  author={Chlebus, Bogdan and G{\k{a}}sieniec, Leszek and Gibbons, Alan and Pelc, Andrzej and Rytter, Wojciech},
  journal={Distributed computing},
  volume={15},
  number={1},
  pages={27--38},
  year={2002},
  publisher={Springer}
}

@article{Chrobak2002,
  title={Fast broadcasting and gossiping in radio networks},
  author={Chrobak, Marek and Gasieniec, Leszek and Rytter, Wojciech},
  journal={Journal of Algorithms},
  volume={43},
  number={2},
  pages={177--189},
  year={2002},
  publisher={Elsevier}
}

@article{CLEMENTI2003,
  title={Distributed broadcast in radio networks of unknown topology},
  author={Clementi, Andrea and Monti, Angelo and Silvestri, Riccardo},
  journal={Theoretical Computer Science},
  volume={302},
  number={1-3},
  pages={337--364},
  year={2003},
  publisher={Elsevier}
}

@inproceedings{Bu2021,
  title={Wireless Broadcast with short labels},
  author={Bu, Gewu and Potop-Butucaru, Maria and Rabie, Mika{\"e}l},
  booktitle={International Conference on Networked Systems},
  pages={146--169},
  year={2020},
  organization={Springer}
}

@inproceedings{Ellen2020,
  title={Constant-length labelling schemes for faster deterministic radio broadcast},
  author={Ellen, Faith and Gilbert, Seth},
  booktitle={Proceedings of the 32nd ACM Symposium on Parallelism in Algorithms and Architectures},
  pages={213--222},
  year={2020}
}

@article{CZUMAJ2019,
  title={Communicating with beeps},
  author={Czumaj, Artur and Davies, Peter},
  journal={Journal of Parallel and Distributed Computing},
  volume={130},
  pages={98--109},
  year={2019},
  publisher={Elsevier}
}

@inproceedings{ghaffari2014near,
  title={Near optimal leader election in multi-hop radio networks},
  author={Ghaffari, Mohsen and Haeupler, Bernhard},
  booktitle={Proceedings of the twenty-fourth annual ACM-SIAM symposium on Discrete algorithms},
  pages={748--766},
  year={2013},
  organization={SIAM}
}

@article{Czumaj2018,
  title={Deterministic communication in radio networks},
  author={Czumaj, Artur and Davies, Peter},
  journal={SIAM Journal on Computing},
  volume={47},
  number={1},
  pages={218--240},
  year={2018},
  publisher={SIAM}
}

@article{Optimal_Deterministic_Broadcasting_in_Known_Topology_Radio_Networks,
  title={Optimal deterministic broadcasting in known topology radio networks},
  author={Kowalski, Dariusz and Pelc, Andrzej},
  journal={Distributed Computing},
  volume={19},
  number={3},
  pages={185--195},
  year={2007},
  publisher={Springer}
}

@article{Pelc2015,
  title={Deterministic broadcasting and gossiping with beeps},
  author={Hounkanli, Kokouvi and Pelc, Andrzej},
  journal={arXiv preprint arXiv:1508.06460},
  year={2015}
}

@misc{oeis,
  title={The online encyclopedia of integer sequences},
  author={Cloitre, Benoit},
  editor={Sloane, Neil},
  year={2002},
  url={http://oeis.org}
}

@inproceedings{chang2018energy,
  title={The energy complexity of broadcast},
  author={Chang, Yi-Jun and Dani, Varsha and Hayes, Thomas P and He, Qizheng and Li, Wenzheng and Pettie, Seth},
  booktitle={Proceedings of the 2018 ACM Symposium on Principles of Distributed Computing},
  pages={95--104},
  year={2018}
}

@article{ghaffari2015randomized,
  title={Randomized broadcast in radio networks with collision detection},
  author={Ghaffari, Mohsen and Haeupler, Bernhard and Khabbazian, Majid},
  journal={Distributed Computing},
  volume={28},
  number={6},
  pages={407--422},
  year={2015},
  publisher={Springer}
}

\newpage

\end{document}